\documentclass[10pt,journal]{IEEEtran}
\usepackage{times}

\usepackage{amsmath}
\usepackage{amsfonts}
\usepackage{latexsym}
\usepackage{amssymb}

\usepackage{upref}
\usepackage{theorem}
\usepackage{graphicx}
\usepackage{psfrag}
\usepackage{multirow}
\usepackage{epstopdf}

\usepackage{txfonts}
\usepackage{color}
\usepackage{enumerate}

\usepackage{comment}
\usepackage{soul}
\usepackage[normalem]{ulem}
\usepackage[ruled,noline]{algorithm2e}

\theoremstyle{plain}
\theorembodyfont{\normalfont\slshape}

\newtheorem{thm}{Theorem$\!$}
\newenvironment{theorem}
{\begin{thm}\hspace*{-1ex}{\bf.}}{\end{thm}}

\newtheorem{lem}[thm]{Lemma$\!$}
\newenvironment{lemma}{\begin{lem}\hspace*{-1ex}{\bf.}}{\end{lem}}

\newtheorem{prop}[thm]{Proposition$\!$}
\newenvironment{proposition}{\begin{prop}\hspace*{-1ex}{\bf.}}{\end{prop}}

\newtheorem{cor}[thm]{Corollary$\!$}

\newtheorem{defn}[thm]{Definition$\!$}
\newenvironment{definition}{\begin{defn}\hspace*{-1ex}{\bf.}}{\end{defn}}

\newtheorem{xmpl}[thm]{Example$\!$}
\newenvironment{example}{\begin{xmpl}\hspace*{-1ex}{\bf.}}{\end{xmpl}}

\newtheorem{prob}[thm]{Problem$\!$}

\newcounter{enumrom}
\renewcommand{\theenumrom}{(\roman{enumrom})}

\makeatletter
\renewcommand{\@endtheorem}{\endtrivlist}
\makeatother

\newcommand\blfootnote[1]{%
  \begingroup
  \renewcommand\thefootnote{}\footnote{#1}%
  \addtocounter{footnote}{-1}%
  \endgroup
}

\makeatletter
\renewcommand{\thefigure}{{\@arabic\c@figure}}
\renewcommand{\fnum@figure}{{\bf Figure\,\thefigure}}
\makeatother

\newcommand{\cC}{{\cal C}}
\newcommand{\cD}{{\cal D}}

\newcommand{\cO}{{\cal O}}
\newcommand{\cP}{{\cal P}}

\newcommand{\cR}{{\cal R}}
\newcommand{\cS}{{\cal S}}

\newcommand{\cW}{{\cal W}}

\newcommand{\be}[1]{\begin{equation}\label{#1}}
\newcommand{\ee}{\end{equation}}

\renewcommand{\leq}{\leqslant}

\renewcommand{\geq}{\geqslant}

\newcommand{\Cref}[1]{Co\-ro\-lla\-ry\,\ref{#1}}

\DeclareMathAlphabet{\mathbfsl}{OT1}{cmr}{bx}{it}

\newcommand{\vect}[1]{\boldsymbol{#1}}
\newcommand{\cellevel}{c}
\newcommand{\hsection}{\cS}
\newcommand{\burst}{b}
\newcommand{\cost}{\cC}
\newcommand{\tree}{\cW}

\newcommand{\cellblocksize}{n}

\outer\def\proclaim #1. #2\par{\medbreak
 \noindent{\bf#1.\enspace}{\sl#2\par}%
 \ifdim\lastskip<\medskipamount \removelastskip\penalty55\medskip\fi}

\mathchardef\inn="3232
\renewcommand{\in}{{\,\inn\,}}

\begin{document}

\sloppy

\title{A Constrained Coding Scheme for Correcting Asymmetric Magnitude-$1$ Errors in $q$-ary Channels}


\author{  \textbf{Evyatar Hemo},\emph{ Graduate Student Member, IEEE,} and \textbf{Yuval Cassuto},\emph{ Senior Member, IEEE}\\}

\maketitle




\maketitle

\begin{abstract}
We present a constraint-coding scheme to correct asymmetric magnitude-$1$ errors in multi-level non-volatile memories. For large numbers of such errors, the scheme is shown to deliver better correction capability compared to known alternatives, while admitting low-complexity of decoding. Our results include an algebraic formulation of the constraint, necessary and sufficient conditions for correctability, a maximum-likelihood decoder running in complexity linear in the alphabet size, and upper bounds on the probability of failing to correct $t$ errors. Besides the superior rate-correction tradeoff, another advantage of this scheme over standard error-correcting codes is the flexibility to vary the code parameters without significant modifications.
\end{abstract}

\section{Introduction}
\blfootnote{Evyatar Hemo (email: evyatarhemo@gmail.com) and Yuval Cassuto (email: ycassuto@ee.technion.ac.il) are with the Andrew and Erna Viterbi Faculty of Electrical Engineering, Technion - Israel Institute of Technology, Haifa $3200003$, Israel.}
The advent of multi-level non-volatile memories (NVMs) has introduced a plethora of coding problems on a variety of $q$-ary channels. Specifically, flash-memory based storage with multiple-level cells has become a fertile ground for addressing these coding problems. In the flash technology, a memory level is determined by the amount of electrical charge trapped inside a floating gate transistor. Equivalently, each memory level is represented by a certain value of voltage or current. However, due to a variety of factors such as noise, manufacturing variability, and inter-cell interference, it becomes very challenging to scale the number of levels while keeping the voltage dynamic range fixed. As a result, the gaps between memory levels become narrower, and a growing error rate compromises the data fidelity.
\blfootnote{Part of the results of this paper was presented at the 2015 International Symposium on Information Theory held in Hong-Kong~\cite{ISIT2015}.}
Errors resulting from narrowing margins between levels have been an active subject of research in recent years. In particular, errors with {\em asymmetry} and {\em magnitude limit} have emerged as useful coding models to tackle this problem. Asymmetry of the errors stems from the inherent asymmetry between the write and erase operations in flash, which often renders one error direction much more dominant than the other. The limited magnitudes of the errors come from ``incremental'' error sources such as retention errors (slow decay of charge level) or disturbs (gradual level increments from writes to adjacent cells)~\cite{errors1},\cite{errors2}.

As the density scaling progresses, we expect that {\em asymmetric, magnitude-$1$} errors will be the first to become the dominant reliability issue of multi-level NVMs. Hence our objective is to develop a coding scheme that will deal with a large number of asymmetric magnitude-$1$ errors within a code block. The coding scheme we propose herein will be evaluated in comparison to known coding techniques for asymmetric magnitude-$1$ errors, which we now briefly survey. In the extreme scenario, we can use an {\em all asymmetric magnitude-$1$} error correcting code, which is equivalent to using only the even cell levels out of the $q$ supported levels~\cite{AhlswedeR:02}. A slightly less redundant code, correcting asymmetric magnitude-$1$ errors in up to {\em half} of the code coordinates, can be obtained by the construction of~\cite{asym} when used with the binary repetition code. This option is equivalent to using only even levels or only odd levels within each codeword -- hence termed {\em even/odd} in the sequel. When the code block length is not too large, we may also use the construction of~\cite{asym} with lower-redundancy binary codes, such as {\em BCH codes}, to achieve a better tradeoff between rate and correctability. A generalization of~\cite{asym}, to more general error distributions can be found in~\cite{Yue2}. Finally, another alternative is to use codes coming from the recently developed theory for {\em L1 distance} codes reported in a sequence of papers by Bose, Tallini and others. See for example~\cite{bose} and citations thereof. In the sequel we show that for moderate to high error rates, our coding scheme gives a better correctability-rate tradeoff compared to the existing alternatives.

The proposed scheme, based on constraint coding, encodes data in a way that asymmetric magnitude-$1$ errors will be easily detected and corrected, but with milder restrictions than the all-errors correcting all-even code or the half-errors correcting even/odd code. The scheme is called {\em Non-Consecutive Constraint (NCC)}, attesting to the enforced constraint of disallowing the codeword to contain two adjacent levels from $\{0,\ldots,q-1\}$. The NCC scheme has three main advantages:
\begin{enumerate}
  \item Performance: the NCC outperforms alternative coding schemes for random errors, when the error rate is moderate to high.
  \item Complexity: the decoder of the NCC has low complexity, and achieves the performance of the maximum likelihood (ML) decoder.
  \item Flexibility: the tradeoff between the redundancy of the code and its error correction capabilities can be easily adjusted without re-designing the code or altering the decoder and encoder.
\end{enumerate}
Section~\ref{sec:asym_errors} describes common coding schemes for asymmetric errors, Section~\ref{sec:NCC} presents the NCC coding scheme and its algebraic formulation. It is an important fact that the resilience of a codeword against asymmetric magnitude-1 errors depends only on its {\em histogram} (how many code-positions have each $q$-ary symbol), and thus the algebraic study of the code works in the domain of word histograms. A maximum-likelihood, low complexity decoding algorithm for the NCC is presented in Section~\ref{sec:dec2}, and detailed performance analysis appears in Section~\ref{sec:perfrom}. Finally, finite block-length analysis for the corresponding information-theoretic model -- the $q$-ary Z-channel -- is presented in Section~\ref{sec:model}, and showing the good performance of the proposed coding scheme.


\section{Asymmetric Errors with Magnitude 1} \label{sec:asym_errors}
As described in the introduction, we assume throughout the paper that all errors are asymmetric with magnitude $1$, in the \emph{downward} direction. Obviously, when restricting the error pattern to asymmetric errors with magnitude $1$, using standard symmetric-error $q$-ary error correcting codes is wasteful. A general recipe for converting a standard error-correcting code to a higher-rate version of the code designed for asymmetric errors with magnitude $1$ can be found in~\cite{asym}. The idea behind~\cite{asym} in a nutshell is that in order to protect an $\cellblocksize$-cell $q$-ary memory block from asymmetric errors with magnitude $1$, we can convert the $q$-ary symbol to its binary form, encode only the least significant bits (LSBs) with a standard binary error correcting code, and map back to the $q$-ary form. By using this method, it is possible to use well-known codes and achieve very high-rate coding schemes.
When using this recipe with the binary repetition code we get a $q$-ary code where in each codeword either all levels are even or all are odd. This scheme allows the correction of $\left\lfloor \frac{\cellblocksize-1}{2} \right\rfloor$ asymmetric magnitude-$1$ errors. Later in the paper we refer to this coding scheme as the \textbf{even/odd} code.
\begin{example}\label{ex:EvenOdd}
The following $q=8$, $\cellblocksize=4$ words $(2,2,4,0)$ and $(1,7,3,1)$ are legitimate even/odd codewords. However, $(2,2,4,1)$ is not a legitimate codeword since it contains both odd and even values. Notice that if an asymmetric error with magnitude $1$ occurs in the first coordinate of the first codeword we get $(1,2,4,0)$. After reading such a word we can definitely deduce that an error occurred. There are three even values and a single odd value, therefore, by taking the majority it is clear that the error occurred in the first coordinate and the original codeword can be successfully restored.
\end{example}
Another simple coding scheme is the \textbf{all-even} code~\cite{AhlswedeR:02}, which is obtained by using only the even levels from the alphabet, and which can correct $\cellblocksize$ asymmetric magnitude-$1$ errors.
\section{The Non-Consecutive Constraint (NCC)} \label{sec:NCC}
We now present the {\bf Non-Consecutive Constraint} (NCC) coding scheme, which is specifically designed to correct large numbers of asymmetric errors with magnitude $1$.
Let the state of the memory cell be represented as a discrete cell level $\cellevel$, taken from the integer set $\{0,\ldots,q-1\}$.
\begin{definition}
The $NCC(\cellblocksize,q)$ is a constraint in which every $\cellblocksize$-cell, $q$-level memory block does not contain cells with consecutive memory levels.
\label{df:NCC}
\end{definition}
\begin{example}\label{ex:NCC}
The following $q=8$, $\cellblocksize=8$ codeword $(2,5,7,0,2,0,4,4)$ is not an NCC codeword (because $4,5$ are consecutive levels and both appear in the codeword). In contrast, $(2,4,4,0,2,0,4,7)$ is an NCC codeword.
\end{example}
While the definition of the NCC is relatively simple, its error-correcting capabilities are not straightforward, since correctability depends on both the codeword and the error pattern.
\begin{definition}
A \textbf{Maximum-Likelihood Decoder} for the NCC decodes a given word to an NCC codeword that requires minimal number of magnitude-$1$ \emph{upward} corrections.
\end{definition}
\begin{example}\label{ex:NCC2}
Let us take the following $q=8$, $\cellblocksize=10$ NCC codeword $(6,6,6,6,6,2,2,2,2,2)$. Let us assume two of the five $6$'s suffer an asymmetric magnitude-1 error, thus we get the received word of $(5,5,6,6,6,2,2,2,2,2)$. Due to the violation of the NCC, it is easy to notice that this word is erroneous, meaning that either two $6$'s have suffered an error and became $5$, or three $7$'s have suffered an error and became $6$. The first option requires fewer magnitude-$1$ upward corrections, therefore selected. In a similar way, we can correct any other combination of two asymmetric magnitude-$1$ errors. However, in the $\cellblocksize=4$ codeword $(6,6,2,2)$ even a single error cannot be corrected. If one of the $6$'s suffers an error, for example, the erroneous word is $(5,6,2,2)$. In this case the decoder cannot choose between two maximum likelihood NCC codewords $(6,6,2,2)$ and $(5,7,2,2)$, each requiring one magnitude-$1$ upward correction.
\end{example}
\subsection{Information rate}
\begin{definition}
The \emph{information rate} $\cR$ of a code for $\cellblocksize$-cell, $q$-level memory array is defined as
\begin{equation}
\cR=log_q(M)/n, \nonumber
\end{equation}
where $M$ is the number of legal combinations according to the code specification. $\cR$ represents the number of information $q$-ary symbols stored per physical cell.
\label{df:rate}
\end{definition}
\begin{theorem}
The information rate of the $NCC(\cellblocksize,q)$ is given by
\begin{equation}
\cR_{NCC}\left(\cellblocksize,q\right)= \frac{1}{n}\log_q \left[\sum_{k=1}^{\frac{q}{2}} k!\cdot S\left(n,k\right) \cdot \binom{q-k+1}{k} \right],
\label{eq:NCCrate}
\end{equation}
where $S(\cellblocksize,k)$ is the Stirling number of the second kind~\cite{vanLintJ:01}.
\label{th:NCCRate}
\end{theorem}
\begin{proof}
The number of combinations to occupy $k$ non-consecutive levels out of $q$ levels is $\binom{q-k+1}{k}$~\cite{vanLintJ:01}(Ch.13). By multiplying it by the number of surjections from $\cellblocksize$-cell set to $k$-level set we obtain the desired combinations count. The number of surjections from an $n$-set to a $k$-set equals $k!\cdot S\left(n,k\right)$~\cite{Stirling2}. Summing all possible values of $k$ gives exactly the expression inside the $\log$ in~\eqref{eq:NCCrate}.
\end{proof}
In order to compare between the NCC and the even/odd and all-even coding schemes, we now quote the information rates for the other methods.
Due to the fact that for an even/odd codeword all $\cellblocksize$ cells are either all even or all odd, the number of possible combinations is given by $M=2\cdot (q/2)^n$, assuming even $q$. Therefore, the information rate of the $\cellblocksize$-cell, $q$-level even/odd code is given by
\begin{equation}
\cR_{even/odd}\left(\cellblocksize,q\right)= 1-\frac{\cellblocksize-1}{\cellblocksize}\log_q 2.
\label{eq:EvenOddrate}
\end{equation}
The all-even code has an information rate that is independent of $\cellblocksize$, and is equal to $1-log_q 2$. As a result, unlike the NCC and even-odd codes above, it is not possible to change the rate-correction tradeoff of all-even by varying its block size.
Fig.~\ref{fig:rate_vs_n} presents a comparison between the information rates of the three coding schemes as a function of the codeword length $\cellblocksize$. It can be observed that the information rate of the NCC is higher than even/odd and all-even for all values of $\cellblocksize$. However, for higher values of $\cellblocksize$ the difference between the codes gets smaller. The rate hierarchy reflected in Fig.~\ref{fig:rate_vs_n} is no coincidence: the three codes are related through the proper inclusion of their codewords: $\text{all-even}\left(n,q\right) \subset \text{even/odd}\left(n,q\right) \subset NCC\left(n,q\right)$.
\begin{figure}[htbp]
   \centering
   \includegraphics[height=2.1in,keepaspectratio=true,width=0.5\textwidth]{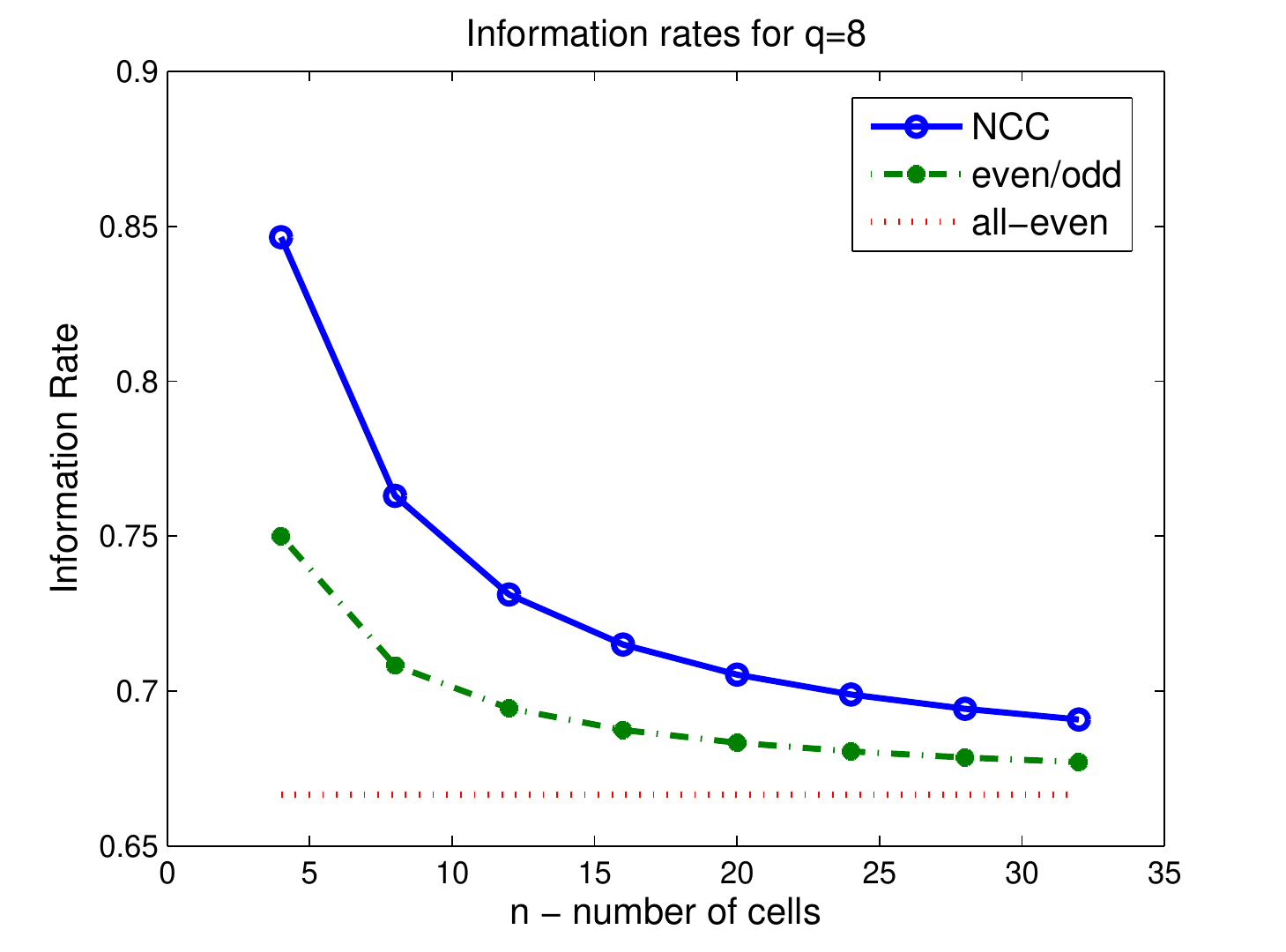}
   \caption{Information rates vs. number of cells per block for NCC (solid - circles) even/odd scheme (dash-dotted - stars) and all-even scheme (dotted) }
   \label{fig:rate_vs_n}
\end{figure} \\
\begin{theorem} \label{th:proofApp}
When $\cellblocksize\rightarrow \infty$ and $q$ is fixed, the NCC, even/odd and all-even schemes have the same information rate.
\end{theorem}
The proof of Theorem~\ref{th:proofApp} appears in Appendix~\ref{app:A}. Even though the NCC has the same rate as the even/odd and all-even schemes in the limit, we show in the sequel that for finite $n$ values it can give better rate/correction tradeoff.
\subsection{Algebraic formulation of the NCC} \label{sec:formalism}
In order to better understand and analyze the NCC coding scheme, we turn to a more formal algebraic definition of the constraint. In the sequel, inequalities involving vectors are interpreted element-wise.
\begin{definition}
Given a vector of cell levels $\vect{c}=(c_1,\ldots,c_n)$, with $c_s\in\{0,\ldots,q-1\}$, define the {\bf histogram vector} as the length-$q$ column vector $\vect{h}$ in which the $i$-th element $i\in\{0,\ldots,q-1\}$ is the number of cells that occupy level $i$ in $\vect{c}$. That is, $h_i=\left| \left\{s\in \{1,\ldots,\cellblocksize\} |c_s=i\right\}\right|$ .
\end{definition}

\begin{example}\label{ex:histogram}
For the following $q=8$, $\cellblocksize=8$ cell-level vector $\vect{c}=(2,2,4,5,2,2,4,0)$, we have $\vect{h}=[1,0,4,0,2,1,0,0]^{T}$.
\end{example}
In the remainder of the section we discuss only histogram vectors and not cell-level vectors, but keeping in mind that there is a many-to-one relation between cell-level vectors and histogram vectors.

\begin{proposition}\label{prop:diag}
Let $\vect{h}$ be a histogram vector of some memory word. $\vect{h}$ is a histogram of a legal NCC codeword if and only if it satisfies
\begin{equation}
diag\left(\vect{h}\right)A\vect{h}=\vect{0},
\end{equation}
where $diag\left(\vect{h}\right)$ stands for a square diagonal matrix with the elements of vector $\vect{h}$ on its main diagonal,
\begin{equation}
A\triangleq \begin{pmatrix} 0 & 1 & 0 & \cdots & 0 \\ 0 & 0 & 1 &\cdots & 0 \\ \vdots & \vdots & \vdots & \ddots & 0\\ 0 & 0 & 0 & \cdots & 1 \\ 0 & 0 &0 &0 & 0  \end{pmatrix},
\end{equation}
and $\vect{0}$ is the all-zero column vector.
\end{proposition}
\begin{proof}
$A\vect{h}$ is the shifted histogram $[h_1,h_2,\ldots,h_{q-1},0]^{T}$, and $diag\left(\vect{h}\right)A\vect{h}$ is the scalar multiplication of the vectors $\vect{h}$ and $A\vect{h}$. Therefore, we get
\begin{equation}
diag\left(\vect{h}\right)A\vect{h}=\left[h_0\cdot h_1,h_1\cdot h_2,\ldots,h_{q-2}\cdot h_{q-1},0\right]^T,
\end{equation}
and requiring this vector to be the all-zero vector is equivalent to having no two consecutive occupied levels.
\end{proof}
Let $\vect{h}$ be a histogram vector of some memory word, and let $\vect{e}$ be an {\em error} histogram vector with all non-negative elements, whose $i$-th element specifies how many memory cells moved from level $i$ to level $i-1$ due to asymmetric magnitude-$1$ errors. The effect of each such error is decrementing $h_i$ and incrementing $h_{i-1}$. As a result, we get that the received histogram $\vect{w}$ is
\begin{equation}
\vect{w}=\vect{h}+T\vect{e},
\end{equation}
where
\begin{equation}
T\triangleq \left(
  \begin{array}{rrrrrr}
    0 & 1 & 0 & 0 & \cdots & 0 \\
    0 & -1 & 1 & 0 & \cdots & 0 \\
    0 & 0& -1 & 1 & \cdots & 0 \\
    \vdots & \vdots & \vdots & \ddots & \ddots & \vdots \\
    0 & 0& 0 & \cdots & -1 & 1 \\
    0 & 0& 0 & 0 & \cdots & -1 \\
  \end{array}
\right)_{q \times q}.
\end{equation}
Note that by definition $\vect{e}$ should satisfy $\vect{e}\leq\vect{h}$. We now define the matrix
\begin{equation}
T^{\dag}\triangleq \left(
  \begin{array}{rrrrrr}
    0 & 0 & 0 & 0 & \cdots & 0 \\
    0 & -1 & -1 & -1 & \cdots & -1 \\
    0 & 0& -1 & -1 & \cdots & -1 \\
    \vdots & \vdots & \vdots & \ddots & \ddots & \vdots \\
    0 & 0& 0 & \cdots & -1 & -1 \\
    0 & 0& 0 & 0 & \cdots & -1 \\
  \end{array}
\right)_{q \times q},
\end{equation}
and observe that it represents the inverse action of $T$ because $T^{\dag}T$ gives the identity mapping on length $q$ vectors of the form $(0,u_1,\ldots,u_{q-1})$, which is exactly the form of error histogram vectors $\vect{e}$.
\begin{definition}
Given two histogram vectors $\vect{h_1},\vect{h_2}$ we define their {\bf downward difference vector} $\vect{d}$ and their {\bf upward difference vector} $\vect{\hat{d}}$ as follows
\begin{gather}
\vect{d}\left(\vect{h_1},\vect{h_2}\right)= \frac{T^{\dag}\left(\vect{h_2}-\vect{h_1}\right)+|T^{\dag}\left(\vect{h_2}-\vect{h_1}\right) |}{2},  \nonumber \\
\vect{\hat{d}}\left(\vect{h_1},\vect{h_2}\right)= -\frac{T^{\dag}\left(\vect{h_2}-\vect{h_1}\right)-|T^{\dag}\left(\vect{h_2}-\vect{h_1}\right) |}{2}.
\end{gather}
\label{df:diff_vecs}
\end{definition}
$|\cdot|$ represents element-wise absolute value. Notice that the elements of $\vect{d}$ and $\vect{\hat{d}}$ are non-negative. In addition, $\vect{d}\left(\vect{h_1},\vect{h_2}\right)-\vect{\hat{d}}\left(\vect{h_1},\vect{h_2}\right)=T^{\dag}\left(\vect{h_2}-\vect{h_1}\right)$. We now move to define a distance function for histograms. In all the discussion of the distance, we assume that the histogram vectors have the same length (the alphabet size $q$), and the same sum (the block size $n$).
\begin{definition}
Given two histogram vectors $\vect{h_1},\vect{h_2}$ we define their {\bf downward distance} $\cD^{\downarrow}$ and their {\bf upward distance} $\cD^{\uparrow}$ as follows
\begin{gather}
\cD^{\downarrow}\left(\vect{h_1},\vect{h_2}\right) = \| \vect{d}\left(\vect{h_1},\vect{h_2}\right)\|,  \nonumber \\
\cD^{\uparrow}\left(\vect{h_1},\vect{h_2}\right) = \| \vect{{\hat{d}}}\left(\vect{h_1},\vect{h_2}\right)\|.
\end{gather}
$\|\cdot\|$ stands for the $\ell_1$ vector norm.
\label{df:distance}
\end{definition}
Note the symmetry properties of the distance functions $\vect{d}\left(\vect{h_2},\vect{h_1}\right)=\vect{\hat{d}}\left(\vect{h_1},\vect{h_2}\right)$, and $\cD^{\downarrow}\left(\vect{h_2},\vect{h_1}\right)=\cD^{\uparrow}\left(\vect{h_1},\vect{h_2}\right)$.
\begin{example}\label{ex:distance}
For the following $q=8$, $\cellblocksize=8$ histograms $\vect{h_1}=[0,2,0,3,3,0,0,0]^{T}$ and $\vect{h_2}=[2,0,0,2,4,0,0,0]^{T}$, the downward and upward difference vectors are $\vect{d}(\vect{h_1},\vect{h_2})=\left[0,2,0,0,0,0,0,0\right]^{T}$ and $\hat{\vect{d}}(\vect{h_1},\vect{h_2})=\left[0,0,0,0,1,0,0,0\right]^{T}$, respectively. The downward and upward distances are $\cD^{\downarrow}(\vect{h_1},\vect{h_2})=2$ and $\cD^{\uparrow}(\vect{h_1},\vect{h_2})=1$, respectively. In other words, in order to get $\vect{h_2}$ from $\vect{h_1}$, it is required to perform $\cD^{\downarrow}(\vect{h_1},\vect{h_2})=2$ downward movements of elements in $\vect{h_1}$ (from level $2$ to $1$) and $\cD^{\uparrow}(\vect{h_1},\vect{h_2})=1$ upward movement (from level $3$ to $4$).
\end{example}
A tool to analyze and design codes is the measure of $t$-confusability between histograms.
\begin{definition}
The histogram $\vect{h_1}$ is \textbf{$t$-confusable} as $\vect{h_2}$ if after suffering $t$ errors (or less) a maximum-likelihood decoder returns $\vect{h_2}$.
\end{definition}

Note that the definition of $t$-confusability does not mean that {\em every} pair of cell-level vectors with the corresponding histograms is $t$-confusable. But it does mean that there exists at least one such pair of cell-level vectors that can confuse the decoder in the case of $t$ errors or less.

\begin{proposition} \label{prop:confuse}
$\vect{h_1}$ is $t$-confusable as $\vect{h_2}$ if and only if the following two conditions hold
\begin{enumerate}
\item $\vect{d}\left(\vect{h_1},\vect{h_2}\right)\leq \vect{h_1}$, $\vect{\hat{d}}\left(\vect{h_1},\vect{h_2}\right)\leq \vect{h_2}$,
\item $\cD^{\uparrow}\left(\vect{h_1},\vect{h_2}\right) \leq \cD^{\downarrow}\left(\vect{h_1},\vect{h_2}\right)\leq t$.
\end{enumerate}
\end{proposition}
\begin{proof}
If $\vect{h_1}$ can be confused as $\vect{h_2}$ after suffering $t$ errors, then there exist two error histogram vectors $\vect{e}\leq \vect{h_1}$ and $\vect{f}\leq \vect{h_2}$ such that $\|\vect{f}\|\leq \|\vect{e}\|\leq t$  that satisfy $\vect{h_1}+T\vect{e}=\vect{h_2}+T\vect{f}$. By reordering and multiplying by $T^{\dag}$ we get that $T^{\dag}\left(\vect{h_2}-\vect{h_1}\right)=\vect{e}-\vect{f}$. It follows from Definition~\ref{df:diff_vecs} that $\vect{d}\left(\vect{h_1},\vect{h_2}\right)=\vect{e}-\vect{c}$, $\vect{\hat{d}}\left(\vect{h_1},\vect{h_2}\right)=\vect{f}-\vect{c}$ for some non-negative vector $\vect{c}$, and conditions 1 and 2 hold by the properties of $\vect{e},\vect{f}$.\\
For the other direction we assume that the conditions are met and examine the vectors $\vect{w}=\vect{h_1}+T\vect{d}\left(\vect{h_1},\vect{h_2}\right)$ and $\vect{w}'=\vect{h_2}+T\vect{\hat{d}}\left(\vect{h_1},\vect{h_2}\right)$. By subtracting the equations and multiplying by $T^{\dag}$ we get that $T^{\dag}(\vect{w}'-\vect{w})=\vect{0}$. We can conclude that $\vect{w}'=\vect{w}$ because the only solutions to the linear system $T^{\dag}\vect{x}=\vect{0}$ are of the form $[x_0,0,\ldots,0]$, and necessarily $x_0=0$ because $\vect{w}$ and $\vect{w'}$ sum to the same value $n$. Hence the received histogram $\vect{w}=\vect{w'}$ will confuse the decoder when $\vect{h_1}$ is stored and the error is $\vect{d}\left(\vect{h_1},\vect{h_2}\right)$.
\end{proof}
The following proposition specializes the $t$-confusability conditions to the case where the histograms are NCC histograms.
\begin{proposition}\label{prop:hee}
Let $\vect{h}$ and $\vect{g}$ be NCC histograms, and denote $\vect{e}=\vect{d}\left(\vect{h},\vect{g}\right)$, $\vect{f}=\vect{\hat{d}}\left(\vect{h},\vect{g}\right)$. If $\vect{h}$ is $t$-confusable as $\vect{g}$ then for every $i$
\begin{gather}
h_i\cdot\left[\left(e_i-e_{i-1}\right)-\left(f_i-f_{i-1}\right)\right]+
h_{i-1}\cdot\left[\left(e_{i+1}-e_{i}\right)-\left(f_{i+1}-f_{i}\right)\right]+\nonumber \\
\left[\left(e_i-e_{i-1}\right)-\left(f_i-f_{i-1}\right)\right]\cdot
\left[\left(e_{i+1}-e_{i}\right)-\left(f_{i+1}-f_{i}\right)\right]=0,\label{eq:element}
\end{gather}
where $h_i,e_i,f_i$ are the $i$-th elements of the histogram vectors $\vect{h},\vect{e},\vect{f}$, respectively.
\end{proposition}
\begin{proof}
Given that $\vect{h}$ and $\vect{g}$ satisfy the two conditions of Proposition~\ref{prop:confuse}, it is clear that $\vect{h}+T\vect{e}=\vect{g}+T\vect{f}$. In addition, given that $\vect{g}$ is also a histogram of an NCC codeword, it must also fulfill Proposition~\ref{prop:diag}, giving
\begin{equation} \label{eq:tmp_hee}
diag\left(\vect{h} + T\left(\vect{e} - \vect{f} \right)\right)A\left(\vect{h} + T\left(\vect{e} - \vect{f} \right)\right)=\vect{0}.
\end{equation}
Solving~\eqref{eq:tmp_hee} element-wise gives the expression in~\eqref{eq:element}.
\end{proof}
The following proposition shows how the algebraic framework can be applied to prove correctability results.
\begin{proposition} \label{cor:error}
Let $h_i$ be the $i$-th element of an NCC histogram $\vect{h}$. If $ h_i>2t$ for each occupied level $i$, then $\vect{h}$ is not $t$-confusable as any NCC histogram.
\end{proposition}
\begin{proof}
Let us assume that $\vect{h}$ is $t$-confusable as some other NCC histogram $\vect{g}$. Then by Proposition~\ref{prop:hee} the equation~\eqref{eq:element} is satisfied for every $i$. Recall that $\vect{e}=\vect{d}\left(\vect{h},\vect{g}\right)$ and $\vect{f}=\vect{\hat{d}}\left(\vect{h},\vect{g}\right)$; we now show that~\eqref{eq:element} cannot be satisfied under the conditions of this proposition. There exists a level $i$ with $e_i> 0$, and by condition 1 in Proposition~\ref{prop:confuse} we also get that $h_i > 0$. By the NCC property, $h_{i-1}=0$ and $h_{i+1}=0$, which imply $e_{i-1}=e_{i+1}=0$. In addition, by the same NCC property on $\vect{g}$, either $f_{i-1}=f_{i+1}=0$ or $f_{i}=0$ (or both). Therefore, assigning these values in~\eqref{eq:element} gives two sets of equations
\begin{equation} \label{eq:ele_prf1}
h_i\cdot\left(e_i-f_i\right)=\left(e_i-f_i\right)\cdot \left(e_i-f_i\right)
\end{equation}
when $f_{i-1}=f_{i+1}=0$, or
\begin{gather}
h_i\cdot\left(e_i+f_{i-1}\right)=\left(e_i+f_{i-1}\right)\cdot \left(e_i+f_{i+1}\right) \Rightarrow \nonumber \\
h_i = e_i+f_{i+1} \label{eq:ele_prf2}
\end{gather}
when $f_{i}=0$. Note that the two equations become the same if both $f_{i-1}=f_{i+1}=0$ and $f_{i}=0$.
Given that $\vect{e}\neq \vect{f}$, there exists at least one $i$ for which ~\eqref{eq:ele_prf1} implies $h_i=e_i-f_i$, a contradiction because $h_i>2t$. In addition, we have $\|\vect{f}\|\leq \|\vect{e}\|\leq t$ which implies that $f_{i+1}\leq t$, $e_i\leq t$, therefore,~\eqref{eq:ele_prf2} cannot be satisfied as well. Hence we proved a contradiction showing that $\vect{h}$ has no $t$-confusable histograms.
\end{proof}
\emph{Remark:} Proposition~\ref{cor:error} can be further refined for the special case in which $h_i=2t$. In order to fulfill~\eqref{eq:ele_prf2}, we must have $e_i=f_{i+1}=t$. That means all $t$ errors occurred in level $i$ and a $t$-confusable histogram can only be created by moving the occupation of level $i$ to level $i+1$. This scenario can hold if and only if $h_{i-2}=h_{i+2}=0$.
\begin{example}
Let us calculate the histogram patterns that are $1$-confusable as some other histograms, by using Proposition~\ref{prop:hee}. Let us assume an error occurs in level $j$, hence $e_j=1$ and all the other values of $\vect{e}$ are $0$. We wish to find all possible values of $\vect{h}$ that are $1$-confusable as some other NCC histogram. When $e_j=1$ we know that $f_j=0$ (the error in the confusing histogram must be elsewhere), and using~\eqref{eq:ele_prf2} we get three possible solutions
\begin{enumerate}
\item $h_j=2$, $f_{j+1}=1$ (by \eqref{eq:ele_prf2} with $i=j$). Then from the fact that $f_i=0$ for $i \neq j+1$ (a single error), we also get $h_{j-2}=0$ (by \eqref{eq:element} with $i=j-1$), and $h_{j+2}=0$ (by \eqref{eq:element} with $i=j+2$).
    \item $h_j=1$, $f_{j+1}=0$ (by~\eqref{eq:ele_prf2} with $i=j$), and $h_{j-2}=1$ (by \eqref{eq:element} with $i=j-1$ and $f_{j-1}=1$).
   \item $h_j=1$, $f_{j+1}=0$ (by~\eqref{eq:ele_prf2} with $i=j$), and $h_{j-2}=0$ (by \eqref{eq:element} with $i=j-1$ and $f_{j-1}=0$).
\end{enumerate}
The solutions are graphically depicted in Fig~\ref{fig:histsOneError}. Pattern A in Fig~\ref{fig:histsOneError} corresponds to the first solution, pattern B to the second one, and pattern C to the third solution.
\begin{figure}[htbp]
   \centering
   \includegraphics[height=2.1in,keepaspectratio=true,width=0.5\textwidth]{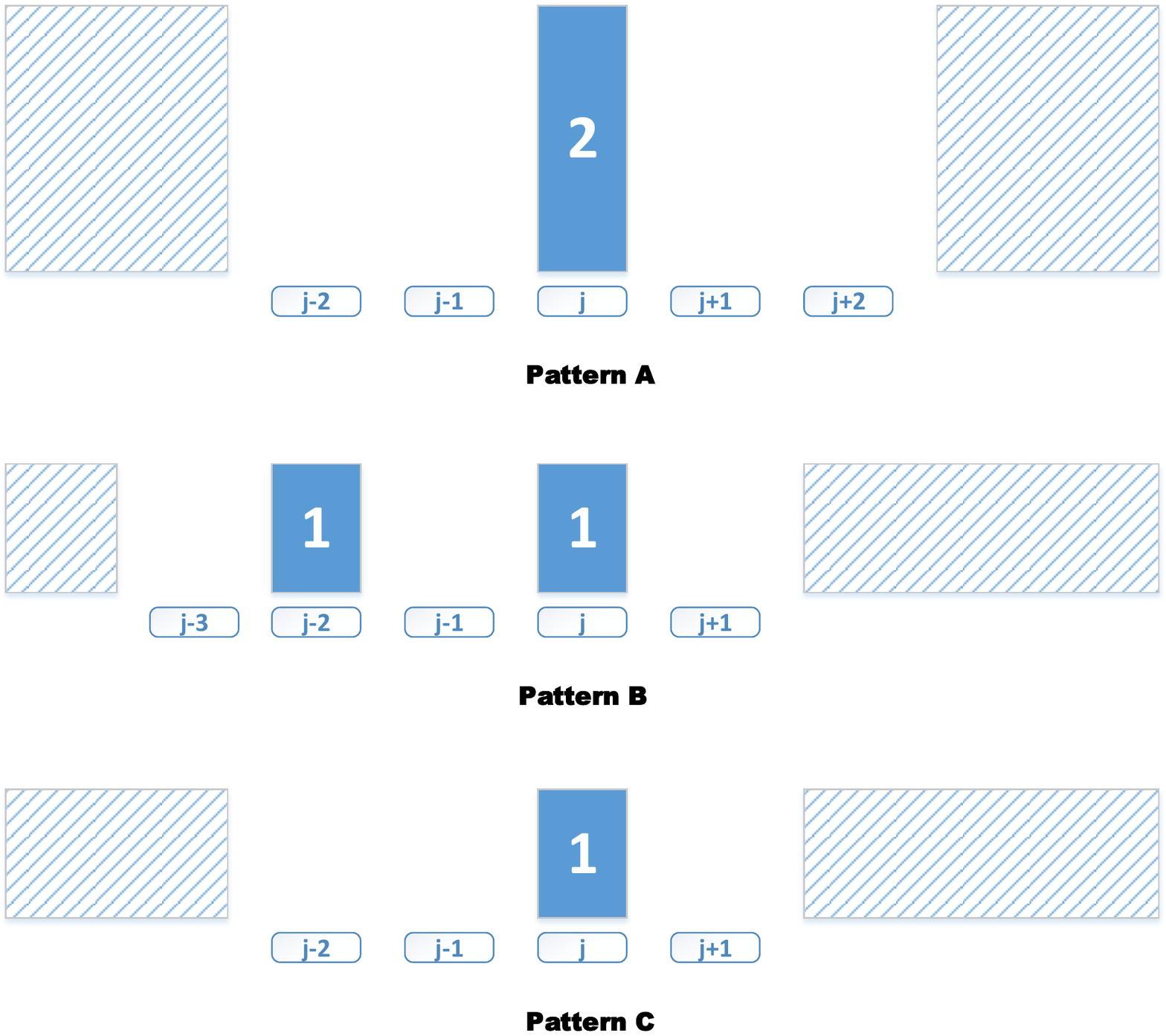}
   \caption{NCC histogram patterns that are 1-confusable as another NCC histogram. These histograms have at least one uncorrectable single error. Dashed bars represent don't-care assignment.}
   \label{fig:histsOneError}
\end{figure}
\end{example}
\section{Decoding \& Encoding Algorithms for the NCC}\label{sec:dec2}
\subsection{A Decoding Algorithm}
In this subsection we describe a decoding algorithm for the NCC. The objective is a maximum-likelihood decoding algorithm that finds the nearest codeword to the received memory word, in the sense that a minimal number of $+1$ cell transitions map between the memory word and the output codeword. The strengths of the proposed decoding algorithm are that it guarantees to return the nearest codeword, is invariant to the code block length $\cellblocksize$, and has very low complexity of $\cO\left(q\right)$. To present the algorithm, we first give several formal definitions.
\begin{definition}
For a histogram vector $\vect{h}$, we define a \textbf{section} $\hsection\left(\vect{h}\right)$ as a minimal-length sub-vector of $\vect{h}$ that is separated from other occupied levels by at least two empty levels. $N_{\hsection}\left(\vect{h}\right)$ is the number of sections in $\vect{h}$.
\label{df:section}
\end{definition}
In the sequel, sections will be used by the algorithm as parts of the histogram that can be decoded independently.
\begin{definition}
A \textbf{burst} $\burst\left(\vect{h}\right)$ is defined as a maximal length sub-vector in a histogram $\vect{h}$ with all its entries greater than zero.  $N_{\burst}\left(\vect{h}\right)$ is the number of bursts in $\vect{h}$. 
\label{df:burst}
\end{definition}
$\hsection_l \left(\vect{h}\right)$ represents the $l$-th section of a certain histogram vector $\vect{h}$. We drop the $\vect{h}$ notation when it is clear to which histogram vector the section refers to. $\burst_j \left(\hsection_l\right)$ represents the $j$-th burst in section $\hsection_l$, and for simplicity of notation, we denote $\burst_j \left(\hsection_l\right)\equiv \burst_j^l$.
\begin{example}\label{ex:section_burt}
For the following $q=12$, $\cellblocksize=15$ histogram vector $\vect{h}=\left[2,0,0,1,3,1,0,0,1,2,0,5\right]^{T}$, there are three sections: $\hsection_1\left(\vect{h}\right)=\left[2\right]$, $\hsection_2\left(\vect{h}\right)=\left[1,3,1\right]^{T}$ and $\hsection_3\left(\vect{h}\right)=\left[1,2,0,5\right]^{T}$. The sections $\hsection_1$ and $\hsection_2$ contain only one burst each, and $\hsection_3$ has two bursts: $\burst_1^3=\burst_1\left(\hsection_3\right)=\left[1,2\right]^{T}$ and $\burst_2^3=\burst_2\left(\hsection_3\right)=\left[5\right]$.
\end{example}
From Definition~\ref{df:burst}, it is obvious that any burst of two levels or more must violate the NCC. Therefore, the purpose of the decoding algorithm is to resolve all the NCC violations expressed as a set of bursts. As a consequence, a decoded codeword must not contain a burst of two levels or more. An NCC violation in a burst can be resolved by up to two options: either all the burst occupations of odd levels must be moved up by one to the even ones, or all the burst occupations of even levels must be moved up to the odd ones (a graphic illustration is presented in Fig.~\ref{fig:hists_ex1}). As will be explained further on, for some scenarios, only one of the movement options is possible.
\begin{figure}[htbp]
   \centering
   \includegraphics[height=2.1in,keepaspectratio=true,width=0.5\textwidth]{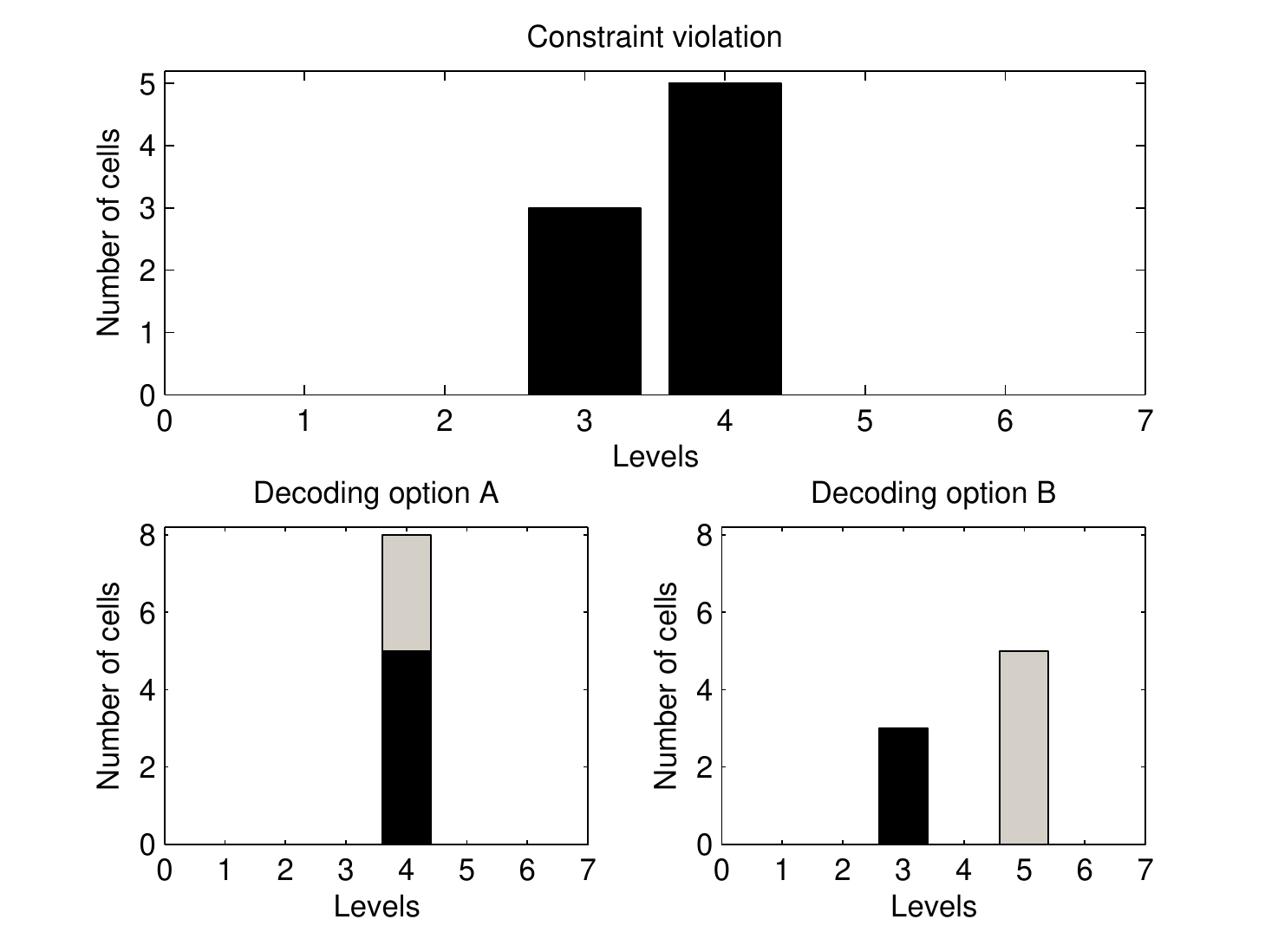}
   \caption{Constraint violation and two options for decoding: moving up the even level or moving up the odd level. Corrected cells are in gray.}
   \label{fig:hists_ex1}
\end{figure}
\begin{definition} \label{df:sigma}
We define the {\bf movement} that resolves the NCC violation induced by burst $\burst$ as either $\sigma\left(\burst\right)$, when keeping the highest level of $\burst$ intact, or $\bar{\sigma}\left(\burst\right)$, when the occupation of the highest level of $\burst$ is moved up. 
\end{definition}
\begin{definition} \label{df:Cost2}
For a movement of levels $\pi(\burst)$, the \textbf{cost} of this movement $\cost\left(\pi(\burst)\right)$ is defined as the number of cells whose level is moved in this action. When $\pi(\burst)$ includes the movement of cells occupying the level $q-1$ (which is impossible), we define $\cost\left(\pi(\burst)\right)=\infty$.
\end{definition}
During the decoding process, the different combinations for resolving bursts are determined in order to achieve minimal cost, hence, finding the nearest NCC codeword.
\begin{example}
For $q=10$, $\cellblocksize=12$, let us assume we read a word with the histogram $\vect{h}=\left[0,4,2,0,0,1,0,0,3,2\right]^{T}$. From the existence of two-level bursts, it is clear that this memory word is corrupt by some errors. All three sections contain only one burst each; the burst of the first section is $\burst_1^1=\left[4,2\right]^{T}$. So in order to resolve this violation, it is possible to perform a $\bar{\sigma}\left(\burst_1^1\right)$ movement (which means we move the $2$ cells at level $2$ upward) and obtain $\left[4,0,2\right]^{T}$. Alternatively, we can perform a $\sigma\left(\burst_1^1\right)$ movement (moving the $4$ cells at level $1$ upward), and obtain $\left[0,6,0\right]^{T}$. The costs of these two movement options are $2$ and $4$, respectively, so in order to minimize the cost, it is obvious that $\bar{\sigma}\left(\burst_1^1\right)$ is preferred. The burst of the second section is $\burst_1^2=\left[1\right]$. In case this cell endured an error, it cannot be decoded successfully. The burst of the last section $\burst_1^3=\left[3,2\right]^{T}$ has only one correction option $\sigma\left(\burst_1^3\right)$, because level $q-1$ cannot move upward. The decoded histogram is therefore $\vect{h'}=\left[0,4,0,2,0,1,0,0,0,5\right]^{T}$.
\end{example}
The decoding processes can work on each section independently, but multiple bursts within a section need to be decoded jointly. The decoder is thus specified below through the main iteration on sections in Algorithm~\ref{alg:decode_hist}, and with the decoding of an individual section as a dynamic-programming (Viterbi-like) subroutine in Algorithm~\ref{alg:decode_section}.
\begin{algorithm}\label{alg:decode_hist}
\SetInd{.03in}{.03in}
\SetKwInOut{Input}{input}
\SetKwInOut{Output}{output}
\caption{DecodeHistogram}
\Input{$\vect{h}$}
\Output{$\vect{h'}$}

for $l=1$ to $N_{\hsection}\left(\vect{h}\right)$

\qquad $actionVec=$DecodeSection($\hsection_l\left(\vect{h}\right)$)

\qquad correct $\vect{h}\Rightarrow \vect{h'}$ by applying $actionVec$ on all $\burst_j \in \hsection_l\left(\vect{h}\right)$

end

\end{algorithm}
 The algorithm starts with the burst, and for every burst $j$ and a movement $\pi \in \left\{\sigma,\bar{\sigma}\right\}$, it builds a concatenated cost table $\tree_j\left(\pi\right)$, which holds the minimal cost of movements for resolving all bursts up to burst $j$. The algorithm also stores the movement path $\Pi_j$ which holds the sequence of movements with the optimal cost.
 During the burst iteration we distinguish between two cases: in a burst $j$ with odd length, not moving the high level of the burst necessitates not moving the high level of burst $j-1$, and in a burst $j$ with even length, {\em moving} the high level of burst $j$ necessitates not moving the high level of burst $j-1$. The decoding process ends with the highest burst that determines which movement is possible for the previous burst (which was already solved). It is clear that the decoding algorithm has $\cO\left(q\right)$ complexity and uses two tables of size $2\sum_{l=1}^{N_{\hsection}}N_{\burst}\left(\hsection_l\right)$, which is smaller than $2q$.

\begin{algorithm}[ht]\label{alg:decode_section}
\SetInd{.03in}{.03in}
\SetKwInOut{Input}{input}
\SetKwInOut{Output}{output}
\caption{DecodeSection}
\Input{$\hsection$}
\Output{$actionVec$}
//\underline{initialization:}

$\tree_1\left(\sigma\right)=\cost\left(\sigma\left(\burst_1\right)\right)$, $\tree_1\left(\bar{\sigma}\right)=\cost\left(\bar{\sigma}\left(\burst_1\right)\right)$

$\Pi_1\left(\sigma\right)=\Pi_1\left(\bar{\sigma}\right)=\emptyset$

//\underline{tables fill}

for $j=2$ to $N_{\burst}\left(\hsection\right)$

\qquad $\alpha = argmin_{\left\{\sigma,\bar{\sigma}\right\}}\left(\tree_{j-1}\left(\sigma\right),\tree_{j-1}\left(\bar{\sigma}\right)\right)$

//\underline{for odd-length bursts}:

\qquad $\tree_j\left(\sigma\right)=\cost\left(\sigma\left(\burst_j\right)\right)+ \tree_{j-1}\left(\sigma\right)$

\qquad $\Pi_j\left(\sigma\right) = \left[\Pi_{j-1}\left(\sigma\right),\sigma\right]$

\qquad $\tree_j\left(\bar{\sigma}\right)=\cost\left(\bar{\sigma}\left(\burst_j\right)\right)+ \min\left(\tree_{j-1}\left(\sigma\right),\tree_{j-1}\left(\bar{\sigma}\right)\right)$

\qquad $\Pi_j\left(\bar{\sigma}\right) = \left[\Pi_{j-1}\left(\alpha\right),\alpha\right]$

//\underline{for even-length bursts}:

\qquad $\tree_j\left(\sigma\right)=\cost\left(\sigma\left(\burst_j\right)\right)+ \min\left(\tree_{j-1}\left(\sigma\right),\tree_{j-1}\left(\bar{\sigma}\right)\right)$

\qquad $\Pi_j\left(\sigma\right) = \left[\Pi_{j-1}\left(\alpha\right),\alpha\right]$

\qquad $\tree_j\left(\bar{\sigma}\right)=\cost\left(\bar{\sigma}\left(\burst_j\right)\right)+ \tree_{j-1}\left(\sigma\right)$

\qquad $\Pi_j\left(\bar{\sigma}\right) = \left[\Pi_{j-1}\left(\sigma\right),\sigma\right]$









end

//\underline{making final decision}

$\alpha = argmin_{\left\{\sigma,\bar{\sigma}\right\}}\left(\tree_{N_{\burst}\left(\hsection\right)}\left(\sigma\right),\tree_{N_{\burst}\left(\hsection\right)}\left(\bar{\sigma}\right)\right)$


$actionVec = \left[\Pi_{N_{\burst}}\left(\alpha\right),\alpha\right]$
\end{algorithm}

\subsection{An Encoding Algorithm}\label{sec:encoder}
In this section we present an encoder for the NCC. The purpose of the encoder is to map a decimal value $x$ to an $n$-cell, $q$-level NCC codeword. The main idea behind the encoder is translating the analytic expression counting NCC codewords~\eqref{eq:NCCrate} to a constructive mapping. We start with the following definitions.
\begin{definition}
Let $LUT$ be a look-up table whose $q/2$ elements are calculated by
\begin{equation}
LUT\left(k\right) = \sum_{i=1}^{k} i!\cdot S\left(n,i\right) \cdot \binom{q-i+1}{i},
\end{equation}
where $1\leq k \leq \frac{q}{2}$.
\label{df:LUT0}
\end{definition}
\begin{definition}
Let $\lambda\left(m,k,j\right)$ be an enumeration function mapping an integer $j$ out of $1,\ldots,\binom{m}{k}$ to a subset of $k$ distinct indices from $\{0,\ldots,m-1\}$.
\label{df:LT1}
\end{definition}
\begin{example}
For $m=4$ and $k=2$, the value of $j$ is between $1$ to $\binom{4}{2}=6$, and corresponds to one of the following sets: $\left\{1,2\right\},\left\{1,3\right\},\left\{1,4\right\},\left\{2,3\right\},\left\{2,4\right\},\left\{3,4\right\}$. Therefore, e.g. $\lambda\left(4,2,1\right)=\left\{1,2\right\}$ and $\lambda\left(4,2,6\right)=\left\{3,4\right\}$.
\end{example}
\begin{definition}
Let $\pi\left(k,i\right)$ be an enumeration function mapping an integer $i$ in the interval $1,\ldots,k!$ to a permutation of the indices $1,\ldots,k$.
\label{df:LT2}
\end{definition}
Algorithm~\ref{alg:HistogramPartition} presents a recursive $\cO\left(n\right)$ algorithm that enumerates partitions counted by the Stirling numbers of the second kind, and maps integers to partitions.
\begin{algorithm}[ht]\label{alg:HistogramPartition}
\SetInd{.03in}{.03in}
\SetKwInOut{Input}{input}
\SetKwInOut{Output}{output}
\caption{StirPar}
\Input{$n,k,x$}
\Output{$\cP$}
//\underline{recursion termination conditions}

if $n=k$

\qquad $\cP = \left[ \left\{c_1\right\},\left\{c_2\right\},\ldots,\left\{c_n\right\} \right]$

elseif $k=1$

\qquad $\cP = \left[ \left\{c_1,c_2,\ldots,c_n\right\} \right]$

end

//\underline{following Stirling's recurrence}

$\tilde{x} = x - k \cdot S\left(n-1,k\right)$

if $\tilde{x} > 0$

\qquad $\tilde{\cP} = StirPar\left(n-1,k-1,\tilde{x}\right) $

\qquad $\cP = \left[\left\{c_n\right\},\tilde{\cP}\right]$

else

\qquad $ \tilde{k} = \left\lceil x / S\left(n-1,k\right) \right\rceil$

\qquad $ \tilde{x} = x - \left(\tilde{k}-1\right) \cdot S\left(n-1,k\right) $

\qquad $ \tilde{\cP} = StirPar\left(n-1,k,\tilde{x}\right) $

\qquad append $c_n$ to $\tilde{k}-th$ element of $\tilde{\cP}$ and return as $\cP$

end

\end{algorithm}
\begin{example}
There are $S\left(5,3\right) = 25$ partitions of an $n=5$ set to $k=2$ non-empty sets. Using Algorithm~\ref{alg:HistogramPartition} gives
$ StirPar \left(5,3,4\right)= \left\{c_4,c_5\right\},\left\{c_3,c_1\right\},\left\{c_2\right\}$, and $StirPar \left(5,3,23\right)=\left\{c_5\right\},\left\{c_1\right\},\left\{c_2,c_3,c_4\right\}$.
\end{example}
We now present the encoding process of $x$ to a length-$n$ NCC codeword by the following steps. Encoding starts with the following calculations:
\begin{enumerate}
\item Let $k$ be the minimal index of LUT for which $x<LUT\left(k\right)$. $k$ represents the number of occupied levels in the histogram of the encoded codeword. If $k>1$ set $y=x-LUT\left(k-1\right)$, otherwise set $y=x$.
\item Let $ i_{\pi}=\left\lfloor y / \left( S\left(n,k\right)\binom{q-k+1}{k} \right)\right\rfloor+1$
\item Let $j_{\lambda} =\left\lfloor \left(y - \left(i_{\pi}-1\right)S\left(n,k\right)\binom{q-k+1}{k} \right)/S\left(n,k\right)\right\rfloor +1$
\item Let $\hat{y} = y - \left(i_{\pi}-1\right)S\left(n,k\right)\binom{q-k+1}{k} - \left(j_{\lambda}-1\right)S\left(n,k\right)+1$
\end{enumerate}
After calculating $k, i_{\pi}, j_{\lambda} ,\hat{y}$ we can now construct the NCC codeword by the following steps.
\begin{enumerate}
\item Let $\vect{s}$ be the length-$k$ vector whose elements are the $k$ elements of $\lambda\left(q-k+1,k,j_{\lambda}\right)$ in ascending order.
\item To each index $i=1,\ldots,k$ in $\vect{s}$, add the value $i-1$.
\item Let $\cP$ be a partition obtained by $StirPar\left(n,k,\hat{y}\right)$.
\item Permute $\cP$ by the permutation obtained by $\pi\left(k,i_{\pi}\right)$.
\item Create the NCC codeword by assigning the cells that appear in the $i$-th set of $\cP$ to the level equaling the $i$-th element of $\vect{s}$.
\end{enumerate}
\begin{example}
Let us assume we wish to encode the value $x=1660$ to an $n=5$, $q=8$ NCC codeword. We start with the preliminary calculations that give:
\begin{enumerate}
\item $LUT=\left[8,638,3638,4838\right]$ therefore $k=3$, we set $y=1660-638=1022$
\item $ i_{\pi}=3$
\item $j_{\lambda} =1$
\item $\hat{y} = 23$
\end{enumerate}
We now construct the codeword by:
\begin{enumerate}
\item $\vect{s}=\lambda\left(6,3,1\right)=\left[0,1,2\right]$.
\item After applying the second step we get $\vect{s}=\left[0,2,4\right]$.
\item $\cP=StirPar\left(5,3,23\right)=\left\{c_5\right\}\left\{c_1\right\}\left\{c_2,c_3,c_4\right\}$.
\item $\pi\left(3,3\right)=\left(3,2,1\right)$, so after permuting $\cP$ we get $\cP=\left\{c_2,c_3,c_4\right\}\left\{c_1\right\}\left\{c_5\right\}$.
\item By pairing $\cP$ with $\vect{s}$ we get the following codeword $\vect{c}=\left(2,0,0,0,4\right)$.
\end{enumerate}
So the decimal value $x=1660$ is encoded to the $n=5$, $q=8$ NCC codeword $\vect{c}=\left(2,0,0,0,4\right)$.
\end{example}
The reverse mapping is a straightforward application of the same steps above in reverse order.
\section{Performance Analysis}\label{sec:perfrom}
In this section we wish to evaluate the error-correcting performance of the NCC code. An exact calculation of the correction probability is challenging in general, so we derive lower bounds. We also show experimentally that the NCC code gives superior performance to known coding alternatives for the same error model.
\subsection{Experimental results}\label{subsec:perf_exp}
Table~\ref{tb:empiric} presents the correction performance of the NCC code for $q=8$. The two left columns include the code length $\cellblocksize$ and its corresponding information rate. The rest of the table presents the probability to fully-correct a specified number of uniformly distributed asymmetric magnitude-$1$ errors.
\begin{table} [ht]
\caption{Full-correction probability for the $n$-cell 8-level NCC}
\centering
\begin{tabular}{|c|c||c|c|c|c|c|c|}
  \hline
   \multicolumn{2}{|c||}{$q=8$} & \multicolumn{6}{|c|}{Total number of errors $\left\|\vect{e}\right\|$} \\
  \hline
  n & $\cR_{NCC} $ & 1 & 2 & 3 & 4 & 5 & 6 \\
  \hline
  5 & 0.816 & 0.801 & 0.478 & 0.170 & 0.043 & 0.007 & 0   \\
  9 & 0.752 & 0.967 & 0.908 & 0.805 & 0.635 & 0.384 & 0.193  \\
  13 & 0.726 & 0.993 & 0.981 & 0.960 & 0.927 & 0.869 & 0.777  \\
  17 & 0.712 & 0.998 & 0.995 & 0.990 & 0.983 & 0.971 & 0.952 \\
  \hline
\end{tabular}  \label{tb:empiric}
\end{table}
As can be seen in Table~\ref{tb:empiric}, the correction capability improves as $\cellblocksize$ increases (but the rate decreases). Moving between rows in the table to different rate-correction tradeoffs is done flexibly without need to change the decoder. Note that there is no guarantee that given a certain number of errors it can always be fully corrected. Whether a certain error can be corrected depends on both the codeword and the error pattern. What does not show in the table is that even in instances that did not fully correct, many of the errors were indeed corrected. In fact, this property will be next shown to result in superior performance of NCC. The way we compare the codes' performance is by measuring the symbol-error rate (SER) at the decoder output for a given rate of asymmetric magnitude-$1$ errors at the input. In Fig.~\ref{fig:empiric1} we plot the output SER for three codes: even/odd, $BCH\left(15,5,7\right)$ and length $7$ NCC, with equal rate of $0.777$ (plus a no-coding option). The quoted BCH code is a binary code that is used in the construction of~\cite{asym} to get asymmetric magnitude-$1$ correction. The x-axis is the input SER, which is the probability that a symbol exhibits an asymmetric magnitude-$1$ error.
\begin{figure}[htbp]
   \centering
   \includegraphics[height=2.1in,keepaspectratio=true,width=0.5\textwidth]{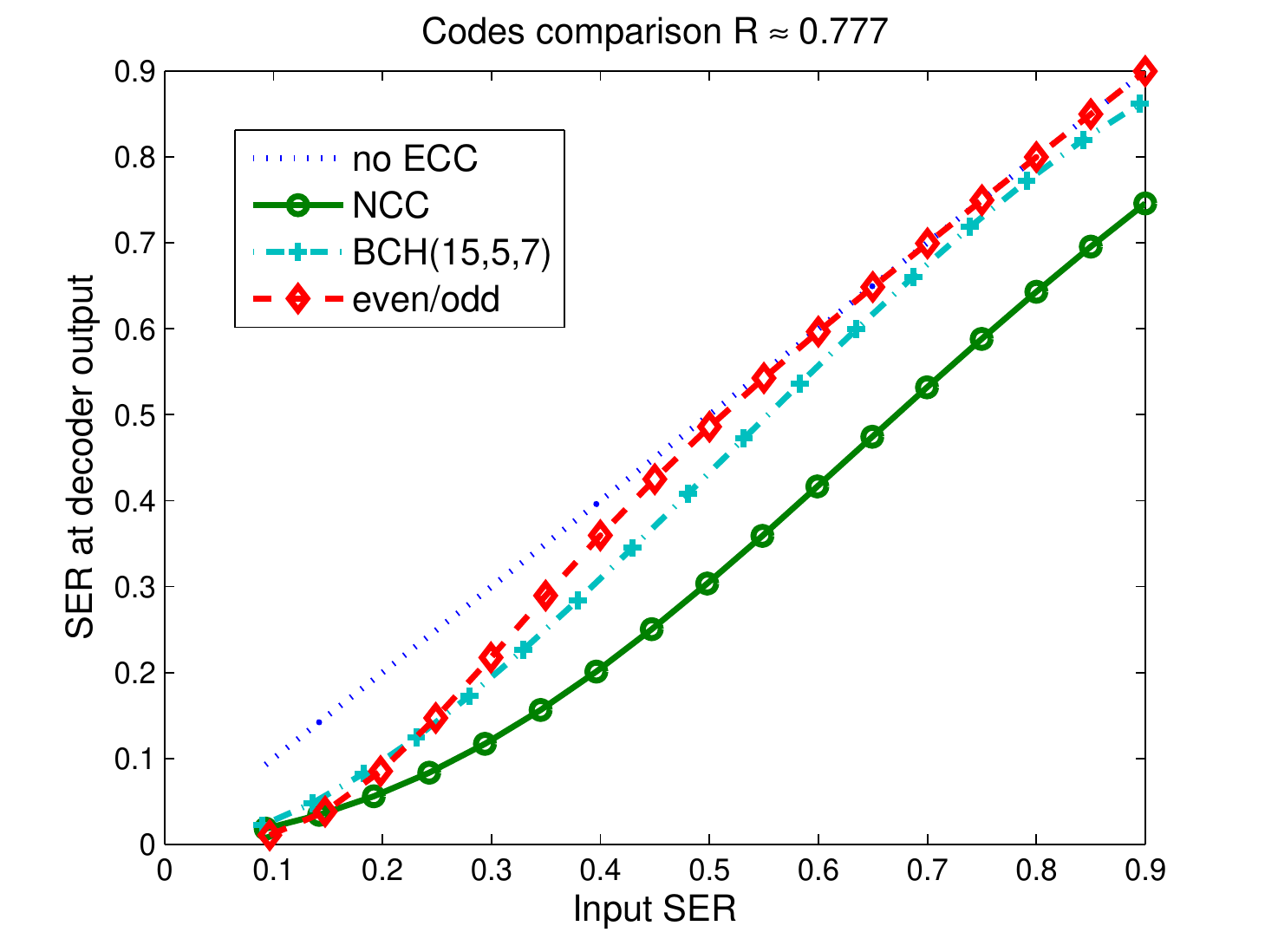}
   \caption{Correcting capabilities as a function of media SER for information rate of 0.777 and $q=8$. NCC (solid \& circles), BCH(15,5,7) (dashed \& diamonds), even-odd (semi-dashed \& pluses). Rate 1 no-ECC option is shown as reference (dotted). }
   \label{fig:empiric1}
\end{figure}
From Fig.~\ref{fig:empiric1} one can notice that the NCC code outperforms all the other codes with equal rate, for input SER values around 0.2 and higher. The SER remaining at the decoder output will be taken care of by a standard ECC, which in the case of NCC will require less redundancy because of the lower residual error rates. Note that the short lengths of the NCC code blocks do not limit the size of the memory word lines. Long word-lines can be supported by concatenating multiple NCC codewords.

\subsection{Bounds on performance}
In this sub-section, we study {\em analytically} the decoding performance of the NCC code. Here the objective is to analyze the block-error performance of the code, that is, its ability to correct {\em all} the errors that corrupted the codeword\footnote{Note that in Section~\ref{subsec:perf_exp} the focus is on the symbol-error and not block-error performance.}. Given $t$, we are interested in the probability $\cP_t$ that the ML decoder {\em fails} to correct all $t$ errors, where the probability is calculated with the uniform distribution on both the error locations and the stored codeword. Calculating this probability exactly seems a hard combinatorial problem, so we derive two upper bounds on $\cP_t$. The first upper bound (Theorem~\ref{th:LB1}) is simpler, albeit less tight. It also enables deriving the asymptotic behavior of the performance more easily. The second upper bound (Theorem~\ref{th:LB2}) is a much tighter bound, but more complicated to derive and analyze.
Let us first make the following definition.
\begin{definition} \label{df:subset_t}
Let $B_t$ be a subset of the NCC code, $B_t \subset NCC(\cellblocksize,q)$, which includes all the codewords whose histogram vector $\vect{h}$ satisfy  $ h_i>2t,\;\forall i: \; h_i \neq 0$.
\end{definition}
By Proposition~\ref{cor:error}, it is clear that for every codeword in $B_t$ enduring $t$ errors, it is guaranteed that all $t$ errors can be corrected.
\begin{theorem} \label{th:LB1}
Given that $t \leq \left\lfloor \frac{n}{5}\right\rfloor+1$ uniformly selected errors occurred in a uniformly selected $NCC(\cellblocksize,q)$ codeword, the ML decoding-failure probability $\cP_t$ is bounded from above by
\begin{equation}
\cP_t \leq  \frac{\sum_{k=1}^{q/2}k!F_{2t+1}\left(n,k\right) \cdot \left[\binom{q-k+1}{k}-1\right] \cdot \frac{2kt}{n} }{M},
\label{eq:tNaiveLB}
\end{equation}
where $M = q^{n\cR_{NCC}\left(\cellblocksize,q\right)}$ is the number of $n$-cell, $q$-level NCC codewords, and
\begin{equation}
F_{r}\left(n,k\right) \triangleq S\left(n,k\right)-S_{r}\left(n,k\right),
\end{equation}
where $S_r(n,k)$ are the $r$-associated Stirling numbers of the second kind~\cite{associated}. (The combinatorial interpretation of $S_r(n,k)$ is the number of partitions of an $n$-set to $k$ subsets, all of which are of cardinality at least $r$.)
\end{theorem}
\begin{proof}
Let us consider the subset $B_t$ from Definition~\ref{df:subset_t}. By Proposition~\ref{cor:error}, it is clear that for every codeword $\vect{c} \in NCC(\cellblocksize,q)$
\begin{gather}
Prob\left(failure \,|\, \vect{c} \in B_t \right)=0, \nonumber \\
Prob\left(failure \,|\, \vect{c} \in \bar{B}_t \right) \geq 0,
\end{gather}
where $\bar{B}_t$ is the codeword subset complementary to $B_t$. For the bound we need to analyze the dependence of decoding failure on the error pattern, in addition to its dependence on $\vect{c}$. For a given codeword in $\bar{B}_t$, some error patterns cause a decoding failure while others do not. Let us define $E\left(\vect{c}\right)$ as the set of all the error patterns $\vect{e}$ which make the decoder fail on codeword $\vect{c}$. Therefore, the total decoding-failure probability is given by
\begin{equation} \label{eq:boundproof0}
Prob\left(failure\right) = Prob\left( \vect{c} \in \bar{B}_t \right) \cdot Prob\left( \vect{e} \in E\left(\vect{c}\right)|\vect{c} \in \bar{B}_t  \right).
\end{equation}
We first calculate an upper bound on $Prob\left( \vect{c} \in \bar{B}_t \right)$ by calculating an upper bound on the size of the subset $\bar{B}_t$. The definition of $B_t$ implies that each codeword in $\bar{B}_t$ contains at least one memory level occupied by $2t$ cells or less. $F_{2t+1}\left(n,k\right)$ is the number of partitions of an $n$-set to $k$ subsets in which at least one subset is of cardinality less than $2t+1$. Recall the expression for counting NCC codewords~\eqref{eq:NCCrate}, replacing $S\left(n,k\right)$ by $F_{2t+1}\left(n,k\right)$ and subtracting a single combination of levels from $ \binom{q-k+1}{k}$ (since the level-set $\{0,2,4,\ldots,2k-2\}$ can always correct $t$ errors regardless of the assignment of $n$ cells to the $k$ levels) gives
\begin{equation}
\left|\bar{B}_t\right| \leq \sum_{k=1}^{q/2}k!F_{2t+1}\left(n,k\right) \cdot \left[\binom{q-k+1}{k}-1\right].
\label{eq:boundproof1}
\end{equation}
To get the upper bound on $Prob\left( \vect{c} \in \bar{B}_t \right)$, all left to do is to divide~\eqref{eq:boundproof1} by $M$.
Next, we calculate an upper bound on the probability $Prob\left( \vect{e} \in E\left(\vect{c}\right)|\vect{c} \in \bar{B}_t \right)$ by upper bounding the size of $E\left(\vect{c}\right)$ for $\vect{c} \in \bar{B}_t$. 
If $\vect{c} \in \bar{B}_t$, by definition there is a level $i$ in the codeword histogram $\vect{h}$ such that $h_i\leq 2t$. If an error vector caused decoding failure to $\vect{c}$ it means that at least $\left\lceil \frac{h_i}{2} \right\rceil$ of the $t$ errors occurred in level $i$. Thus given $i$ and $h_i$, the number of failing error words equals
\begin{equation}
\sum_{j=0}^{\left\lfloor \frac{h_i}{2} \right\rfloor}\binom{h_i}{\left\lceil \frac{h_i}{2} \right\rceil+j} \binom{n-h_i}{t-\left\lceil \frac{h_i}{2} \right\rceil-j}.
\label{eq:given_i_hi}
\end{equation}
If we maximize~\eqref{eq:given_i_hi} over all possible $h_i$ and multiply it by the number $k$ of occupied levels in $\vect{c}$, we get the union-bound
\begin{equation}
|E\left(\vect{c}|k\right)| \leq k \max_{1\leq h_i \leq 2t} \left[\sum_{j=0}^{\left\lfloor \frac{h_i}{2} \right\rfloor}\binom{h_i}{\left\lceil \frac{h_i}{2} \right\rceil+j} \binom{n-h_i}{t-\left\lceil \frac{h_i}{2} \right\rceil-j}\right].
\label{eq:bound_E}\end{equation}
From this bound we can get a bound on the second term in~\eqref{eq:boundproof0} by dividing by the number of error words
\begin{gather}
Prob\left( \vect{e} \in E\left(\vect{c}\right)|\vect{c} \in \bar{B}_t ,k \right) \leq \\ \nonumber k \max_{1\leq h_i \leq 2t} \left[\sum_{j=0}^{\left\lfloor \frac{h_i}{2} \right\rfloor}\binom{h_i}{\left\lceil \frac{h_i}{2} \right\rceil+j} \binom{n-h_i}{t-\left\lceil \frac{h_i}{2} \right\rceil-j}\right]/\binom{n}{t}.
\label{eq:bound_cond_prob}
\end{gather}
In order to get an expression that does not depend on $h_i$, we upper bound the argument of the $\max$ function. For convenience we replace $h_i$ with $h$
\begin{gather}\label{eq:boundproof3}
\sum_{j=0}^{\left\lfloor \frac{h}{2} \right\rfloor}\binom{h}{\left\lceil \frac{h}{2} \right\rceil+j} \binom{n-h}{t-\left\lceil \frac{h}{2} \right\rceil-j} \leq \binom{n-h}{t-\left\lceil \frac{h}{2} \right\rceil} \cdot \sum_{j=0}^{\left\lfloor \frac{h}{2} \right\rfloor}\binom{h}{\left\lceil \frac{h}{2} \right\rceil+j}= \\ \nonumber
=2^{h-1} \binom{n-h}{t-\left\lceil \frac{h}{2} \right\rceil} \leq 2^{h-1} \binom{n-1}{t-\left\lceil \frac{h}{2} \right\rceil} \equiv \phi\left(h \right).
\end{gather}
We now focus on the properties of $ \phi\left(h \right)$
\begin{gather}
\phi\left(h+2 \right) = 2^{h+1} \binom{n-1}{t-\left\lceil \frac{h}{2} \right\rceil-1} = \frac{4\left(t-\left\lceil \frac{h}{2} \right\rceil \right)}{n - t+\left\lceil \frac{h}{2} \right\rceil}  \phi\left(h \right) \leq \\ \nonumber \leq \frac{4\left(t-1 \right)}{n - t+1}  \phi\left(h \right),
\end{gather}
which means that for $t\leq n/5+1$, we get that $ \phi\left(h \right) \geq \phi\left(h+2 \right)$ for every $h$. Therefore, $\phi\left(h \right)$ is a monotonically decreasing function for its odd and even entries separately. It is also easy to verify that under this condition, we also get that $\phi\left(2 \right) \geq \phi\left(1 \right)$. As a result, $ \phi\left(h \right)$ is bounded from above by
\begin{equation}
\phi\left(2 \right) = 2 \binom{n-1}{t-1}.
\end{equation}
Hence, we get that
\begin{equation} \label{eq:boundproof4}
Prob\left( \vect{e} \in E\left(\vect{c}\right) |\, \vect{c} \in \bar{B}_t,k \right) \leq \frac{2k\binom{n-1}{t-1}}{\binom{n}{t}} = \frac{2kt}{n}.
\end{equation}
Using the expression in~\eqref{eq:boundproof0}, an upper bound on the failure probability is obtained by multiplying~\eqref{eq:boundproof4} by~\eqref{eq:boundproof1}, and dividing by $M$.
\end{proof}

For the case in which a single error occurs, it is possible to derive the asymptotic behavior for the failure probability when $n$ is large.
\begin{theorem} \label{th:asym}
For $n$-cell, $q$-level NCC, when $q$ is fixed and $n \rightarrow \infty$, the probability $\cP_{t=1}$ of failing to correct a single error approaches $0$ exponentially with respect to $n$
\begin{equation}
\cP_1 \simeq c \cdot n\left(1-\frac{2}{q}\right)^{n},
\end{equation}
where $c$ is some constant.
\end{theorem}
\begin{proof}
Due to the fact that asymptotically $S\left(n,k\right)\sim \frac{k^n}{k!}$ when $k$ is fixed and $n \rightarrow \infty$, and given the following recurrence identity for the $r$-associated Stirling numbers of the second kind~\cite{howard}
\begin{gather}
S_2\left(n,k\right)=n! \sum_{j=0}^k \frac{(-1)^jS\left(n-j,k-j\right)}{j!\left(n-j\right)!} \nonumber \\
S_3\left(n,k\right)=n! \sum_{j=0}^k \frac{(-1)^jS_2\left(n-2j,k-j\right)}{2^jj!\left(n-2j\right)!}
\end{gather}
we get for $n \rightarrow \infty$ that
\begin{gather}
S_2\left(n,k\right) = S\left(n,k\right) - nS\left(n-1,k-1\right) + \nonumber \\
+\binom{n}{2}\cdot S\left(n-2,k-2\right) +\cO\left(n^3 \left(k-3\right)^{n}\right) \nonumber \\
S_3\left(n,k\right) = S_2\left(n,k\right) - \binom{n}{2}\cdot S_2\left(n-2,k-1\right)+ \nonumber \\
+3\binom{n}{4}\cdot S_2\left(n-4,k-2\right)+ \cO\left(n^5 \left(k-3\right)^{n}\right).
\end{gather}
Note that the terms marked as $\cO\left(\cdot\right)$ functions become negligible for large $n$ because they have smaller bases for the exponent. After substituting the expression for $S_2\left(n,k\right)$ inside $S_3\left(n,k\right)$, we get
\begin{gather}
S_3\left(n,k\right) = S\left(n,k\right) - nS\left(n-1,k-1\right) +\binom{n}{2}\cdot S\left(n-2,k-2\right) + \nonumber \\
+ \cO\left(n^3 \left(k-3\right)^{n}\right)-\binom{n}{2}S\left(n-2,k-1\right)+n\binom{n}{2}S\left(n-3,k-2\right)+ \nonumber \\ +\cO\left(n^4 \left(k-3\right)^{n}\right) +3\binom{n}{4}S\left(n-4,k-2\right)+\cO\left(n^5 \left(k-3\right)^{n}\right).
\end{gather}
Therefore, after replacing the Stirling functions with their asymptotic expressions, and after taking the most dominant $\cO\left(\cdot\right)$ terms, we also get that
\begin{gather}
F_3\left(n,k\right) = S(n,k)-S_3(n,k) = n\left(1+\frac{n-1}{2k-2}\right)\cdot \frac{ \left(k-1\right)^{n-1} }{\left(k-1\right)!}+  \nonumber \\ \frac{n\left(n-1\right)}{2}\left(\frac{(n-2)(n-3)}{4(k-2)^2}-\frac{n}{k-2}-1\right)\cdot \frac{ \left(k-2\right)^{n-2} }{\left(k-2\right)!}+  \nonumber \\
+ \cO\left(n^5 \left(k-3\right)^{n}\right).
\end{gather}
Obviously, when $n \rightarrow \infty$ we get that
\begin{gather}
F_3\left(n,k\right) =  n\left(1+\frac{n-1}{2k-2}\right)\cdot \frac{ \left(k-1\right)^{n-1} }{\left(k-1\right)!}
+ \cO\left(n^4 \left(k-2\right)^{n}\right).
\end{gather}
We then only take the most dominant terms in both the numerator and the denominator of~\eqref{eq:tNaiveLB}. For $t=1$, the most dominant term in the numerator is the $k=q/2$ term of the sum; the most dominant term in the denominator is the $k=q/2$ term of~\eqref{eq:NCCrate} (see also Lemma~\ref{prop:ncc_eq_evenodd2} in Appendix~\ref{app:A}). After some manipulation we get
\begin{gather}
\cP_1 \simeq \frac{\left(\frac{q}{2}\right)!  \frac{n\left(1+\frac{n-1}{q-2}\right)\left(\frac{q}{2}-1\right)^{n-1}}{\left(\frac{q}{2}-1\right)!} \frac{q}{2} \frac{q}{n} +\cO\left(n^3 \left(\frac{q}{2}-2\right)^{n-2} \right)}{\left(\frac{q}{2}\right)^n \left(\frac{q}{2}+1\right)}=\nonumber \\
=\frac{2q^2 n\left(1-\frac{2}{q}\right)^{n}}{\left(q+2\right)\left(q-2\right)^2}+\cO\left(n^3 \left(1-\frac{4}{q}\right)^{n} \right)\simeq c\cdot n\left(1-\frac{2}{q}\right)^{n}.
\end{gather}
\end{proof}
The following upper bound is applicable to higher number of errors $t$, and it is tighter when $n$ is not too large. It is a refinement of the upper bound in Theorem~\ref{th:LB1}, hence, its proof is based on the proof of Theorem~\ref{th:LB1}.
\begin{definition}
Let $\bar{B}_{k,h}$ be the subset of NCC codewords occupying $k$ levels where at least one level is occupied by \emph{exactly} $h$ cells.
\end{definition}
\begin{definition} \label{def:gamma}
Let $\Gamma_h\left(n,k\right) $ be the the number of partitions of an $n$-set to $k$ sets in which at least one of the $k$ sets is of cardinality $h$.
\end{definition}
\begin{lemma} \label{lm:set_exact}
The function $\Gamma_h\left(n,k\right) $ is bounded from above by
\begin{equation} \label{eq:lemmaUB}
\Gamma_h\left(n,k\right) \leq \binom{n}{h} \cdot k! \cdot S \left(n-h,k-1\right).
\end{equation}
\end{lemma}
\begin{proof}
Let us choose such a partition by first assigning one of the levels to $h$ cells, and then partitioning the rest of the $n-h$ cells to the $k-1$ remaining levels. The number of choices for the latter is $(k-1)! \cdot S \left(n-h,k-1\right)$. In the first assignment the chosen level and the $h$ cells can be arbitrary, therefore, we also need to multiply the outcome by $k\cdot\binom{n}{h}$. This selection procedure may over count, because the same $h$ set assigned in the first step may be also assigned in the latter step for a different permutation of the $k$ levels.
\end{proof}
\begin{theorem} \label{th:LB2}
Given that $t < \left\lfloor\frac{n}{2}\right\rfloor$ uniformly selected errors occurred in a uniformly selected $NCC(\cellblocksize,q)$ codeword, the ML decoding-failure probability $\cP_t$ is bounded from above by
\begin{equation} \label{eq:lowerbound2}
\cP_t \leq \frac{\sum\limits_{h=1}^{2t} P\left(n,t,h\right) \sum\limits_{k=2}^{\frac{q}{2}} k \cdot \Gamma_h\left(n,k\right) \left[\binom{q-k+1}{k}-1\right]}{M},
\end{equation}
where $\Gamma_r\left(n,k\right)$ is defined in Definition~\ref{def:gamma}, $M = q^{n\cR_{NCC}\left(\cellblocksize,q\right)}$ is the number of $n$-cell, $q$-level NCC codewords, and
\begin{equation}\label{eq:Pnth}
P\left(n,t,h\right) = \sum_{j=0}^{\left\lfloor \frac{h}{2} \right\rfloor}\binom{h}{\left\lceil \frac{h}{2} \right\rceil+j} \binom{n-h}{t-\left\lceil \frac{h}{2} \right\rceil-j}/\binom{n}{t}.
\end{equation}
\end{theorem}
\begin{proof}
In this proof we further refine the enumeration of NCC codewords beyond the class $\bar{B}_t$ used at the right-hand side of ~\eqref{eq:boundproof0} in Theorem~\ref{th:LB1}. Let us denote the histogram of an NCC codeword $\vect{c}$ with $k$ non-zero elements by $hist\left(\vect{c}\right)=\left[h_1,h_2,\ldots,h_k\right]$. Recall from the proof of Theorem~\ref{th:LB1} that for an uncorrectable error to occur in a level occupied by $h$ cells, $\left\lceil h/2 \right\rceil$ errors or more must occur in that exact level. The probability of this event is $P\left(n,t,h\right)$ in~\eqref{eq:Pnth}, which follows easily from~\eqref{eq:given_i_hi}. We now can bound the failure probability by the union bound
\begin{gather}\label{eq:Pfail_first_bound}
Prob\left(failure\right) \leq \\ \nonumber
\sum\limits_{k=1}^{\frac{q}{2}} \sum_{h_1,\ldots,h_k}
Prob\left(hist\left(\vect{c}\right)=\left[h_1,\ldots,h_k\right] \right) \sum\limits_{i=1}^{k} P\left(n,t,h_i\right).
\end{gather}
We can re-write the inner sum as
\begin{equation}\label{eq:h_no_i}
\sum\limits_{i=1}^{k} P\left(n,t,h_i\right) = \sum\limits_{h=1}^{2t} P\left(n,t,h\right) \cdot N\left(h \in hist\left(\vect{c}\right)\right),
\end{equation}
where $N\left(h \in hist\left(\vect{c}\right)\right)$ counts the number of times $h$ appears in $hist\left(\vect{c}\right)$. A simple upper bound is
\begin{equation}\label{eq:N_bound}
N\left(h \in hist\left(\vect{c}\right)\right)  \leq k \delta\left(h \in hist\left(\vect{c}\right)\right),
\end{equation}
where $\delta(\cdot)$ is the indicator function. Now we can write
\begin{gather}\label{eq:hist_delta}
\sum_{h_1,\ldots,h_k}Prob\left(hist\left(\vect{c}\right)=\left[h_1,\ldots,h_k\right] \right) \delta\left(h \in hist\left(\vect{c}\right)\right) = \\ \nonumber Prob\left( \vect{c} \in \bar{B}_{k,h} \right).
\end{gather}
Now substituting~\eqref{eq:N_bound} in~\eqref{eq:h_no_i}, and then~\eqref{eq:h_no_i} in~\eqref{eq:Pfail_first_bound}, we get after reordering the summation and substituting~\eqref{eq:hist_delta}
\begin{gather}\label{eq:Pfail_last_bound}
Prob\left(failure\right) \leq
\sum\limits_{k=1}^{\frac{q}{2}} \sum\limits_{h=1}^{2t}
Prob\left( \vect{c} \in \bar{B}_{k,h} \right) \cdot k \cdot P\left(n,t,h\right).
\end{gather}
By the same arguments of Theorem~\ref{th:LB1}, we get that
\begin{equation} \label{eq:bound_sec1}
\left|\bar{B}_{k,h}\right| = \Gamma_h\left(n,k\right) \binom{q-k+1}{k}.
\end{equation}
So, in order to calculate $Prob\left( \vect{c} \in \bar{B}_{k,h} \right)$ all left to do is to divide~\eqref{eq:bound_sec1} by the total number of NCC codewords $M$. Note that due to the fact that $t \leq \left\lfloor \frac{n}{2}\right\rfloor$, when only one level is occupied in an NCC codeword, any error can be corrected. As a result, we get that the failure probability is $0$ for these combinations, and the summation in~\eqref{eq:Pfail_last_bound} can start from $k=2$. In addition, similarly to Theorem~\ref{th:LB1}, we here too can subtract $1$ from $\binom{q-k+1}{k}$.
Applying these refinements in~\eqref{eq:Pfail_last_bound} and changing summation order give us the expression for the upper bound in~\eqref{eq:lowerbound2}.
\end{proof}
It is possible to derive a tighter expression for the lower bound in Theorem~\ref{th:LB2}, by using the exact count of $\Gamma_r\left(n,k\right)$ given by~\cite{mark}:
\begin{equation}
\Gamma_r\left(n,k\right) = n! \sum\limits_{j=0}^{\min(k, \lfloor n/r\rfloor)} \frac{(-1)^j}{(n-jr)! j!(r!)^j} S \left(n-jr,k-j\right).
\end{equation}
Fig.~\ref{fig:NCCBounds_t1} and Fig.~\ref{fig:NCCBounds_t2} present the upper bounds on the decoding-failure probability for $t=1$ and $t=2$ errors, respectively, as a function of the block-length $n$. They are compared to the experimental results presented in Section~\ref{sec:perfrom}. As can be seen, both upper bounds approach the experimental values for higher values of $n$ and lower values of $t$. We also may notice that the upper bound of Theorem~\ref{th:LB2} is much tighter than Theorem~\ref{th:LB1} for lower values of $n$.
\begin{figure}[htbp]
   \centering
   \includegraphics[height=2.1in,keepaspectratio=true,width=0.5\textwidth]{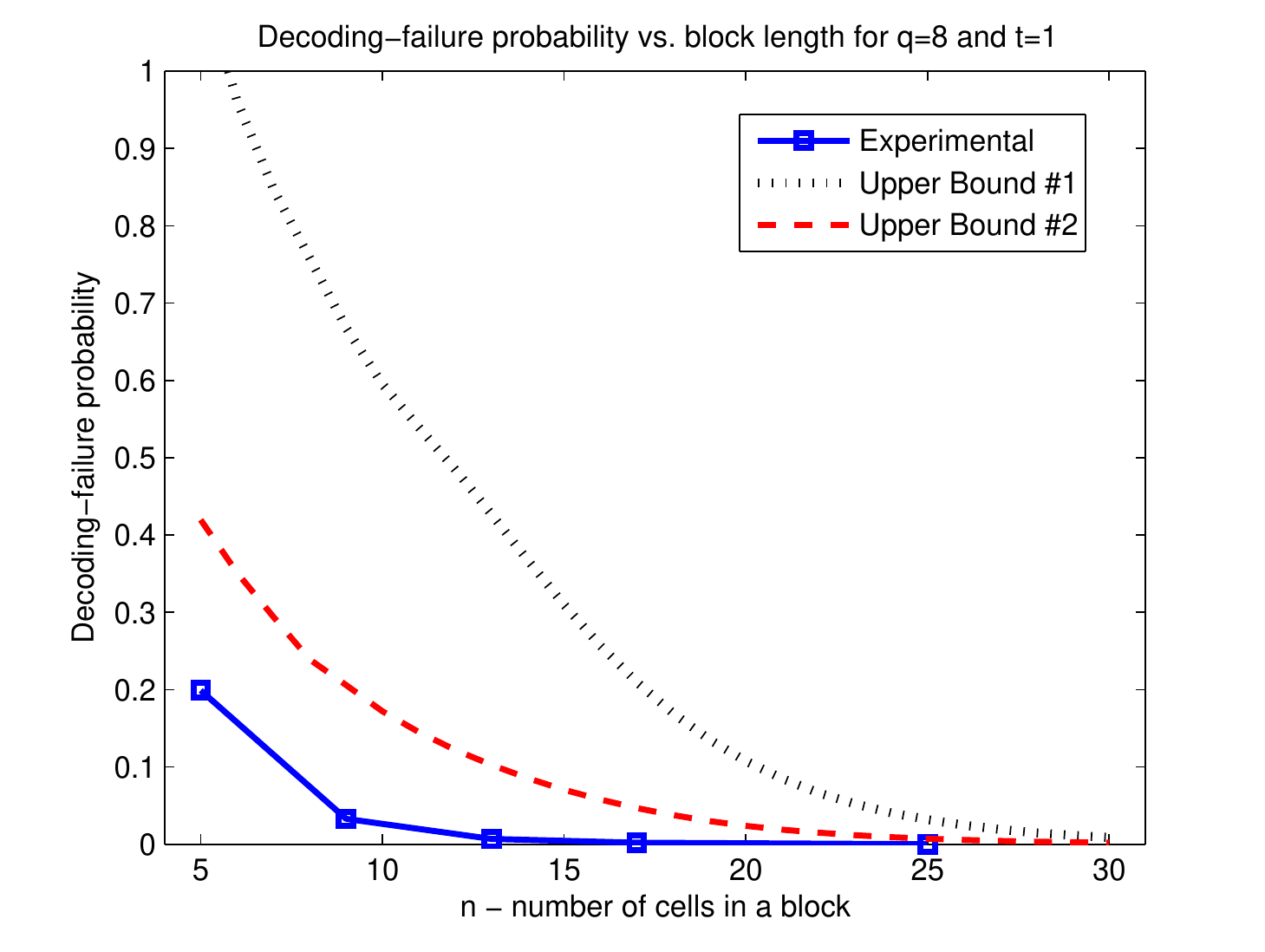}
   \caption{Upper bounds on the probability to fail decoding $1$ error in comparison with experimental results. Solid squared line - experimental results, dotted line - upper bound of Theorem~\ref{th:LB1}, dashed line - upper bound of Theorem~\ref{th:LB2}.}
   \label{fig:NCCBounds_t1}
\end{figure}
\begin{figure}[htbp]
   \centering
   \includegraphics[height=2.1in,keepaspectratio=true,width=0.5\textwidth]{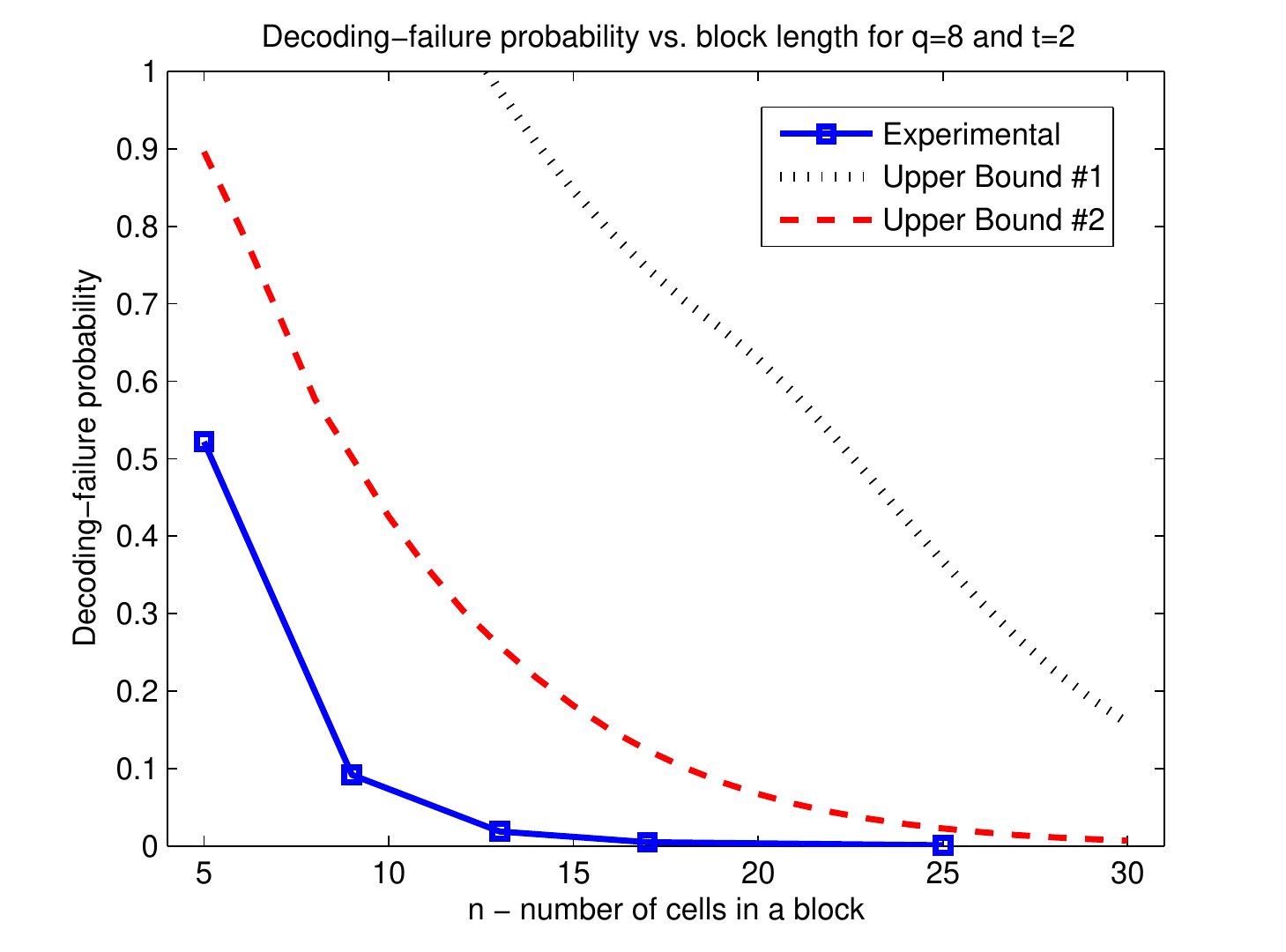}
   \caption{Upper bounds on the probability to fail decoding $2$ errors in comparison with experimental results. Solid squared line - experimental results, dotted line - upper bound of Theorem~\ref{th:LB1}, dashed line - upper bound of Theorem~\ref{th:LB2}.}
   \label{fig:NCCBounds_t2}
\end{figure}

\section{The Q-ary Z-Channel Analysis} \label{sec:model}
\subsection{Capacity and capacity achieving distributions}\label{subsec:z_ch}
A multi-level memory enduring asymmetric, magnitude-1 errors can be modeled by the $q$-ary Z-channel shown in Fig.~\ref{fig:qZcahnnel}. This channel is essentially Shannon's \emph{noisy typewriter} channel, only with a general transition probability $p$ and without the wrap-around transition from $0$ to $q-1$.
\begin{figure}[htbp]
   \centering
   \includegraphics[height=2.1in,keepaspectratio=true,width=0.5\textwidth]{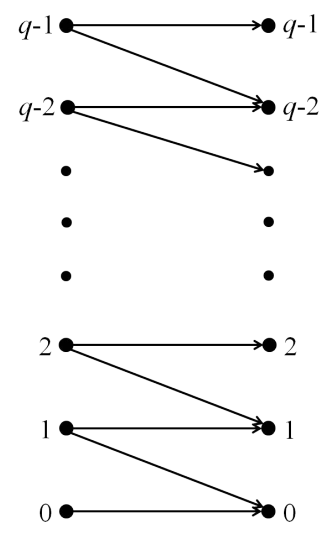}
   \caption{q-ary Z-channel model}
   \label{fig:qZcahnnel}
\end{figure}
A special case of this channel with $q=2$ is the binary Z-channel, which was extensively studied~\cite{Golomb},\cite{Tal}. Both the noisy typewriter channel and the binary Z-channel have analytic expressions for the channel capacity and for the corresponding input distributions. However, analyzing the $q$-ary Z-channel is more complicated, therefore, some numeric solutions must be used~\cite{Shulman}.
The capacity of the generalized probability $p$ noisy typewriter channel (with wrap-around) is given by
\begin{equation}
\mathrm{C} = 1- h_2\left(p\right) \cdot log_q 2,
\end{equation}
where $h_2\left(p\right)$ is the binary entropy function. The capacity is achieved by the uniform input distribution. However, due to the symmetry break in the $q$-ary Z-channel (without wrap-around), calculating the capacity is more complicated, and the mutual information is given by
\begin{gather}
I\left(X;Y\right) = H\left(Y\right) - H\left(Y|X\right)= \\ \nonumber
-\sum_{y} p\left(y\right) log_2 p\left(y\right) -\left(1-p\left(x=0\right)\right)\cdot h_2\left(p\right)
\end{gather}
while $p\left(y\right)$ is given by
\begin{equation} \label{eq:mutual}
p\left(y\right)=\left\{
  \begin{array}{ll}
     p\left(x=0\right)+p \cdot p\left(x=1\right) & y=0 \\
     \left(1-p\right) \cdot p\left(x=q-1\right) & y=q-1 \\
     \left(1-p\right) \cdot p\left(x=y\right)+p \cdot p\left(x=y+1\right) & else \\
  \end{array}
\right.
\end{equation}
We use optimization methods to find the input distribution $p\left(x\right)$ that maximizes the mutual information given in~\eqref{eq:mutual} and achieves the capacity of the channel. Fig.~\ref{fig:InputDist} presents the input distribution (for $q=7$) for several values of the error probability $p$. As can be noticed, when $p$ is very low, the capacity-achieving input distribution is relatively uniform. As $p$ increases, the lowest and highest levels become more probable. When $p$ is relatively high, the probability for the even levels increases and the probability for the odd levels decreases. When $p=0.5$, we get the well known zero-error code for the noisy typewriter channel called the all-even code previously in this paper. This only happens for odd $q$ values; for even $q$ the distribution converges to something non-uniform.\\
\begin{figure}[htbp]
   \centering
   \includegraphics[height=2.1in,keepaspectratio=true,width=0.5\textwidth]{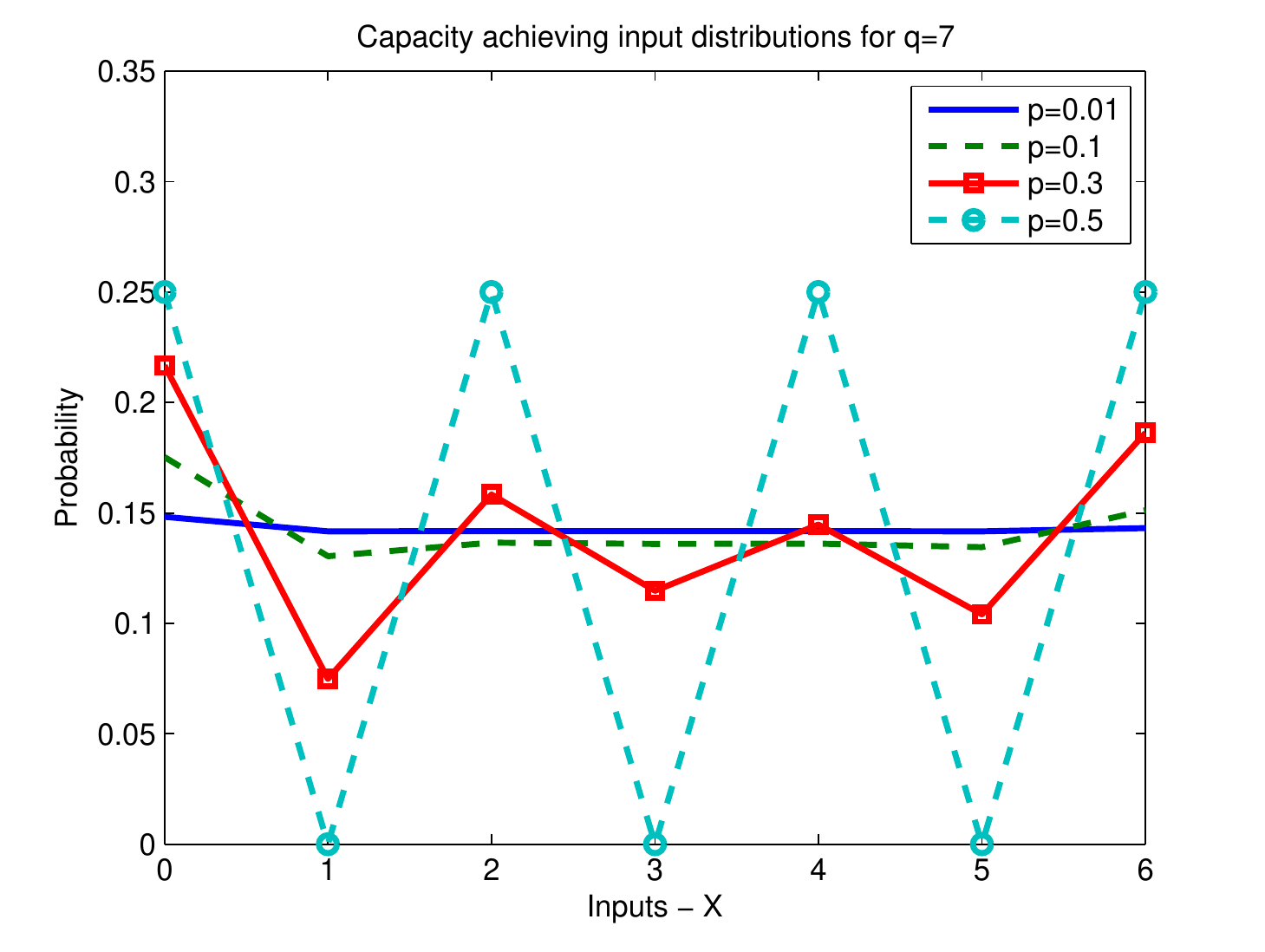}
   \caption{Capacity achieving input distributions for the q-ary Z-channel with q=7 and error probabilities of: $p=0.01$ - solid, $p=0.1$ - dashed, $p=0.3$ - solid \& squares, $p=0.5$ - dashed \& circles.}
   \label{fig:InputDist}
\end{figure}
When we examine the symbol distribution of the NCC code, we see that it follows the capacity-achieving distributions of Fig.~\ref{fig:InputDist} in the following sense. An NCC code with length $n$ has a symbol distribution very close to the capacity-achieving distribution of the $q$-ary Z-channel with some error probability $p$. ``Very close'' means that the mutual information resulting from the NCC distribution is found numerically to be less than $0.01\%$ lower than the capacity. Note that this fact does {\em not} imply that the length $n$ NCC code achieves capacity for the corresponding probability $p$ channel. But Fig.~\ref{fig:PerfImprov} does give an indication that length $n$ NCC codes are well tuned to the corresponding probability $p$ channels. In the figure we plot for three lengths of NCC codes, $n=5,7,10$, the improvement in block-error rate (over uncoded) for each error probability $p$. On each plot we mark by a dot the $p$ parameter whose capacity-achieving distribution is closest to the NCC distribution. It can be seen that the best empiric error-correction performance of the NCC is achieved very close to these $p$ values found by the information-theoretic analysis.\\
\begin{figure}[htbp]
   \centering
   \includegraphics[height=2.1in,keepaspectratio=true,width=0.5\textwidth]{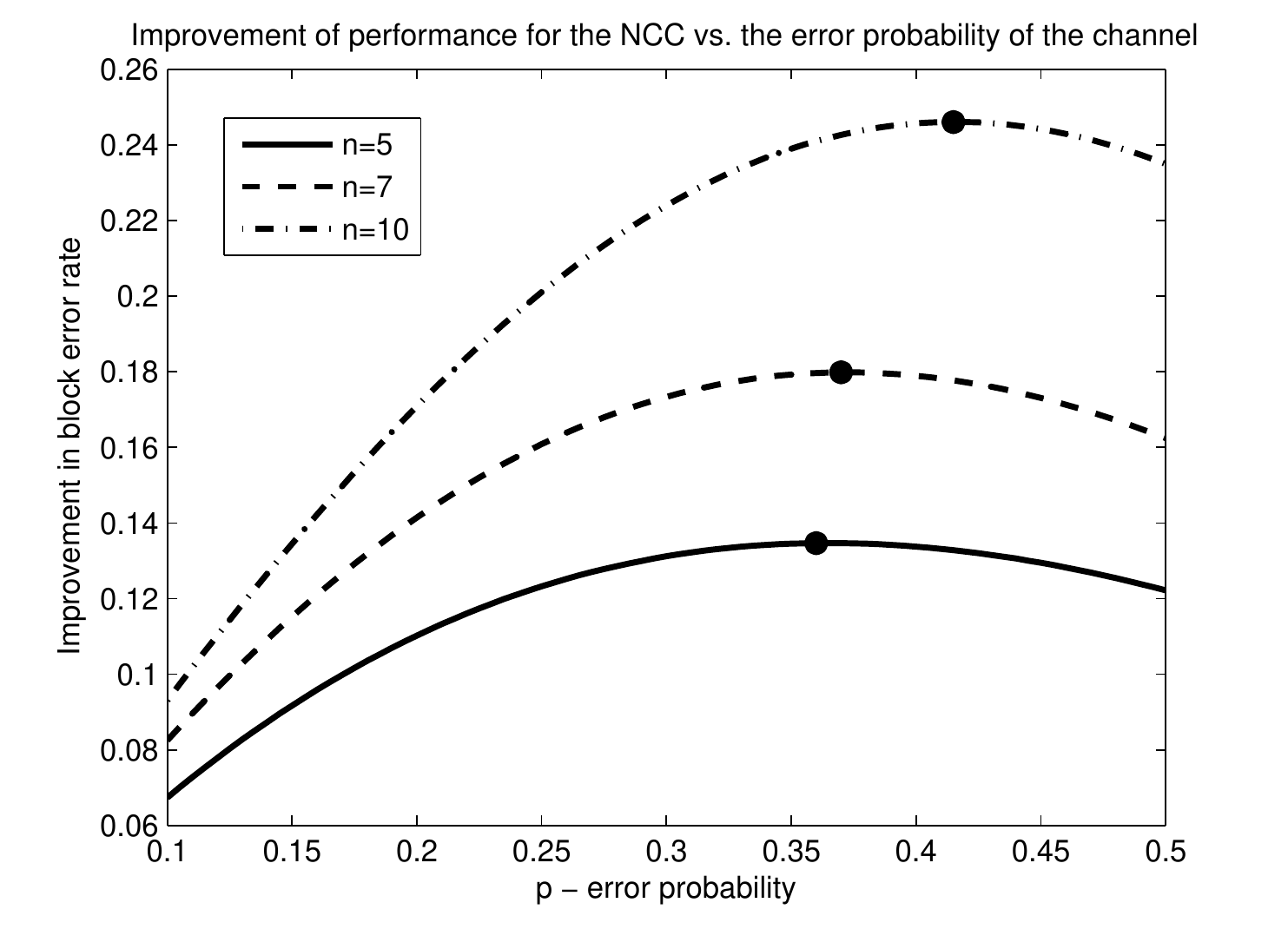}
   \caption{Block error rate improvent of the NCC vs. the error probability of the channel $p$. Solid line: $n=5$, dashed line: $n=7$, dash-dotted line: $n=10$. The dot on each plot marks the channel parameter $p$ for which the NCC distribution is closest to the capacity-achieving distribution.}
   \label{fig:PerfImprov}
\end{figure}
\subsection{Finite block-length analysis}\label{subsec:FBL}
NCC codes are attractive to use for finite lengths $n$ that depend on the channel parameter $p$. In the asymptotic regime where $n$ tends to infinity, it was shown in Section~\ref{sec:NCC} that the NCC code converges to the even/odd code. Therefore, in order to evaluate the performance of the NCC code, we need to use finite block-length analysis tools. In the evaluation we use known results from the well-developed finite block-length information theory. In particular, we use the finite block-length upper bound (converse bound) first introduced by Strassen~\cite{Strassen}, and further refined and extended by Polyanskiy, Poor and Verdu~\cite{Verdu}, Tomamichel and Tan~\cite{Tom}, and by Moulin~\cite{Moulin}.
We start with some formal definitions~\cite{Verdu},\cite{Tom}.
\begin{definition}
A $\left(n,M,\varepsilon\right)-code$ code is a length-$n$ code with $M$ codewords, for which there exists a decoder whose block error probability is smaller than $\varepsilon$.
\end{definition}
The maximal code size achievable with block-length $n$ and error probability $\varepsilon$ is denoted by
\begin{equation}
M^*\left(n,\varepsilon\right) = \max \left\{ M: \exists \: a \, \left(n,M,\varepsilon \right)-code \right\}. \nonumber
\end{equation}
Given that the capacity achieving input distribution is unique, the dispersion of the channel is defined as
\begin{equation}
V \triangleq
  \begin{array}{ll}
     Var\left[i\left(X,Y\right)| X\right],
   \end{array}
\end{equation}
where $i\left(X,Y\right)$ is the information density given by
\begin{equation}
i\left(X,Y\right) = \log \frac{P_{Y|X}\left(y|x\right)}{P_Y\left(y\right)}.
\end{equation}
\begin{theorem} \cite{Strassen}-\cite{Moulin} \label{converse}
For every discrete memoryless channel (DMC) and $\varepsilon$ with $V \geq 0$, the following bound applies
\begin{equation}\label{eq:converse}
\log M^*\left(n,\varepsilon\right)\leq nC - \sqrt{nV}\Phi^{-1}\left(\varepsilon\right) + \frac{1}{2}\log n + \cO\left(1\right),
\end{equation}
where $C$ is the Shannon capacity of the channel, and $\Phi \left(\cdot\right)$ is the Gaussian cumulative distribution function.
\end{theorem}
The channel dispersion $V$ of the $q$-ary Z-channel can be calculated using the capacity-achieving input distributions found by optimization tools in Section~\ref{subsec:z_ch}.
\begin{figure}[htbp]
   \centering
   \includegraphics[height=2.1in,keepaspectratio=true,width=0.5\textwidth]{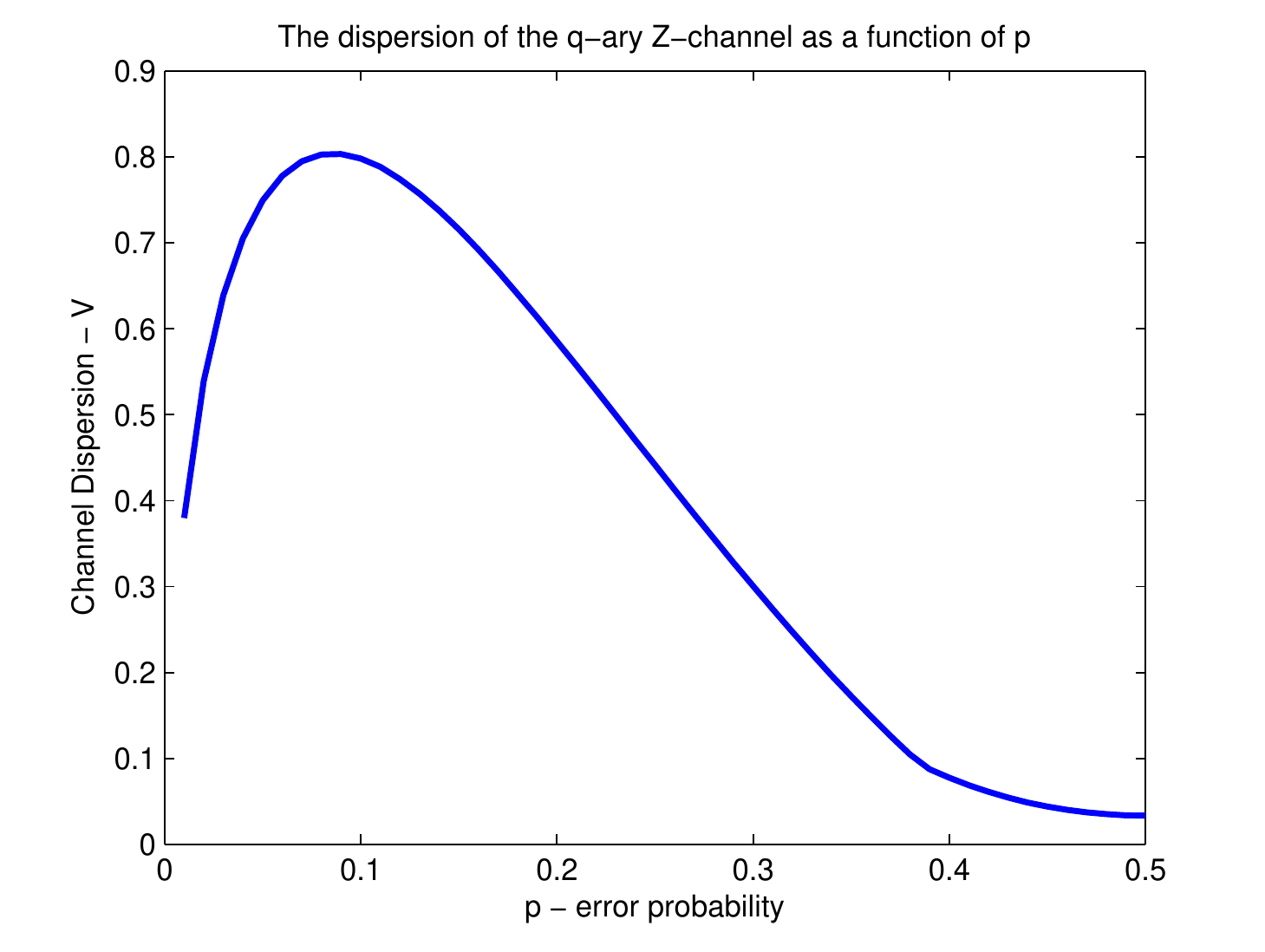}
   \caption{Channel dispersion $V$ of the $q$-ary Z-channel as a function of the error probability $p$.}
   \label{fig:dispersion}
\end{figure}
Fig.~\ref{fig:dispersion} presents the channel dispersion as a function of the error probability of the channel. After finding $V$,~\eqref{eq:converse} can be fully calculated.
Comparing the NCC scheme with Theorem~\ref{converse} is not straightforward and several remarks must be made. First, the fourth element of~\eqref{eq:converse} may not be negligible for small values of $n$. Second, as can be seen in Table~\ref{tb:empiric}, changing the block-length of the NCC inherently changes the error correction capability. Hence, the $n$ and $\varepsilon$ parameters are conjugated in the NCC scheme - meaning that for a given $q$-ary Z-channel with error probability $p=0.1$, using different block-lengths yields different $\varepsilon$'s. So, the way we compared the NCC to Theorem~\ref{converse} is as follows. We fixed the error probability of the channel to $p=0.1$. Then we calculated (by exhaustive simulations) the resulting $\varepsilon$'s for four NCC block-lengths: $7,9,13,17$ (calculating it for higher values of $n$ demands serious computational resources). The calculated $\varepsilon$'s are $[0.0686,0.0407,0.0144,0.0054]$, respectively. Then, we calculated~\eqref{eq:converse} for the four different values of $\varepsilon$.
\begin{figure}[htbp]
   \centering
   \includegraphics[height=2.1in,keepaspectratio=true,width=0.5\textwidth]{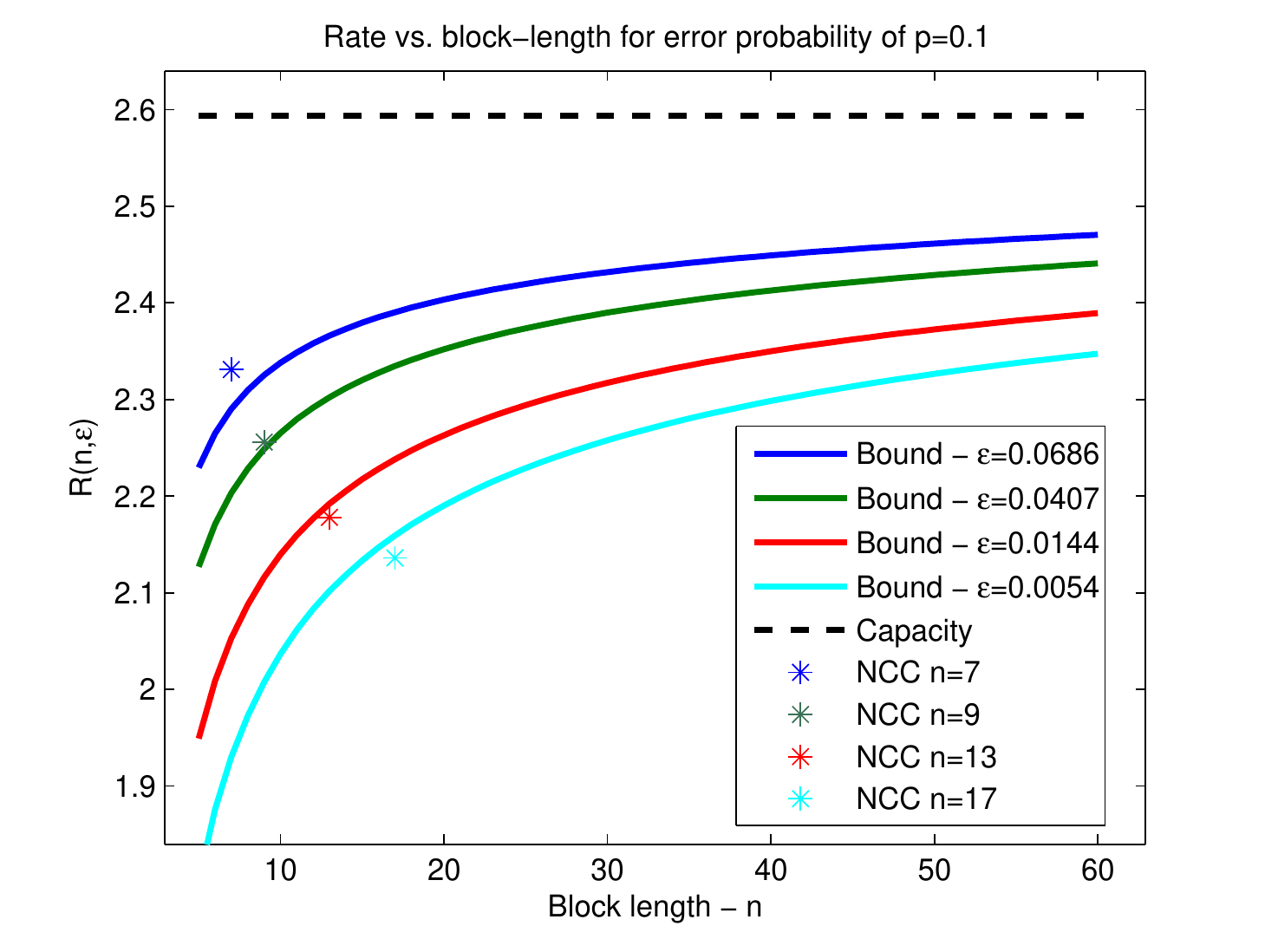}
   \caption{Solid lines - converse bounds for four values of $\varepsilon$, $[0.0686,0.0407,0.0144,0.0054]$. Stars - rates of the NCC for four values of $n$, $[7,9,13,17]$. Dashed horizontal line - the capacity of the channel for $p=0.1$.}
   \label{fig:Converse}
\end{figure}
Fig.~\ref{fig:Converse} presents the comparison between the performance of the NCC and the converse bounds. As can be seen, for low values of $n$, the NCC beats the converse bound, which is possible because of the $\cO\left(1\right)$ term in~\eqref{eq:converse}. As $n$ increases, we can see that the NCC nicely follows the bound, reaching rates that are $<1\%$ lower than the converse bounds. Therefore, the NCC code has the potential to be very close to optimal finite-block length for moderate block-length.

\section{The NCC as an ECC Scheme in Flash Memories} \label{sec:flash}
Approaching the end of this paper, we want to venture beyond the theoretical discussion, and project the NCC scheme on the space of real flash storage devices. Given the attractiveness of the NCC in correcting asymmetric magnitude-1 errors (demonstrated in Fig.~\ref{fig:empiric1}), it is natural to ask how to deploy it in real flash devices. In terms of computational complexity, we have provided in Section~\ref{sec:dec2} efficient algorithms for encoding and decoding NCC codewords. Thus there are two main issues left to resolve:
\begin{enumerate}
\item Coding over block lengths induced by large flash pages.
\item Correcting the residual errors at the decoder output.
\end{enumerate}
It turns out that accomplishing item 1 is very easy, while item 2 is less immediate and is also a topic for further research.

To accommodate large flash pages (also called {\em wordlines} in the flash jargon), recall that the NCC block length $n$ determines both the code rate (decreasing with $n$) and the correction capability (improving with $n$). Typical block lengths shown to be useful/advantageous are from $n=5$ to $n\approx 30$. If we have flash pages of size $N \gg n$, then we simply concatenate $\lceil N/n \rceil$ parallel NCC codewords in each flash page, where the choice of $n$ is done based on the SER estimate and correction specifications for the device. A pictorial illustration is shown in Fig.~\ref{fig:block}. If the parameter $n$ is correctly chosen, each codeword will decode correctly, and a correct full page will be delivered to the user. By that we decouple the codeword length $n$ from the page size $N$.
\begin{figure}[htbp]
   \centering
   \includegraphics[height=2.1in,keepaspectratio=true,width=0.5\textwidth]{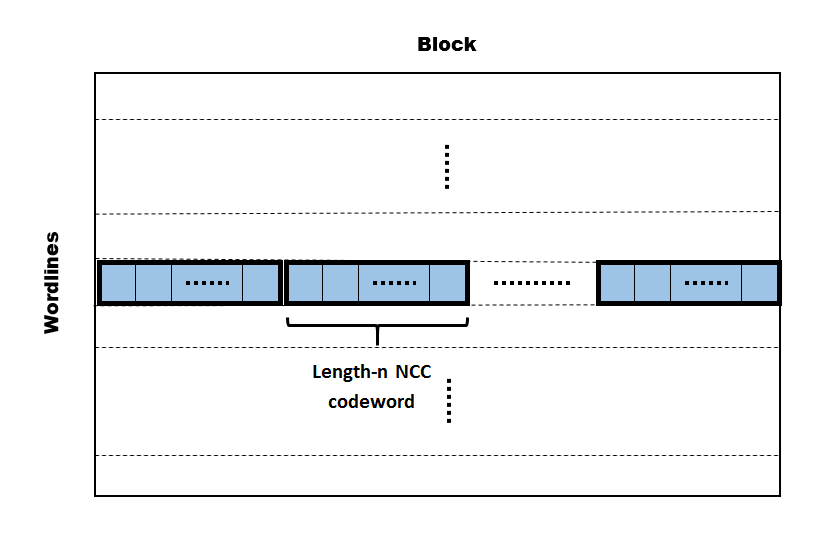}
   \caption{ A memory block consisting of pages (wordlines), where each page includes concatenated length-$n$ NCC codewords.}
   \label{fig:block}
\end{figure}

To deal with residual errors at the NCC decoder output, we can concatenate it with an outer code for symmetric errors such as Reed-Solomon (RS) or LDPC code. If we use a RS code with alphabet size $M$ (the number of NCC codewords), then by standard concatenation an NCC decoding error will translate to a single symbol error in the RS code. Low decoding error probability can be achieved by the concatenated code so long that the NCC decoding-error probability (Table~\ref{tb:empiric}) matches the correction capability of the RS code. For example, if we take an NCC code with $n=13$ on input SER $0.095$, then the resulting output SER is $0.0021$, which means that most of the NCC blocks will be decoded successfully, and then a high-rate RS code will suffice as the outer code. The more challenging aspect of residual errors is the treatment of secondary error types (e.g. the ones treated in~\cite{Yue1}) in conjunction to the dominant asymmetric magnitude-1 errors. If these secondary errors are frequent enough to affect a significant fraction of NCC blocks, then the suggested concatenation with RS codes will not work, since many RS symbols would have residual errors left by the NCC decoder (which was not designed to combat these errors). In such cases, an alternative is to concatenate the NCC with an outer binary code, e.g. an LDPC code, and iterate information between the LDPC and NCC decoders such that the most likely codeword is found. That way, the redundancy from the LDPC code can help the NCC decoder prefer a more likely codeword, even if not all errors are asymmetric magnitude-1. Devising such a decoder is a topic for future work.

In the following example we demonstrate the potential benefit of concatenating an NCC code with an outer LDPC code, such that the symbols left in error by the NCC decoder (whose number is measured as SER in Fig.~\ref{fig:empiric1}) are corrected by the LDPC code. The example shows that the SER reduction in asymmetric magnitude-1 errors offered by the NCC code allows using a low complexity LDPC decoder for the outer code, while alternative asymmetric magnitude-1 codes leave high output SER that cannot be corrected by an LDPC code of reasonable redundancy and complexity. We note that this example is for a quantitative illustration only, and not yet a working scheme, because it assumes that a single residual asymmetric magnitude-1 error maps to a single bit error at the LDPC decoder input. This assumption is not readily true, because of the non-trivial mapping between outer-code bits and NCC codewords (shown in Section~\ref{sec:encoder}).

\begin{example}
The technical requirement for reliability of most flash memories stands on a bit-error rate (BER) of $10^{-15}$. Let us assume a $q=8$ flash memory suffers asymmetric magnitude-$1$ errors such that its raw (input) SER is $0.24$. The raw SER of $0.24$ is mapped\footnote{We approximate a 1-level shift as a 1-bit error, as given e.g. by gray coding.} to raw BER of approximately $0.08$, which is the SER level divided by the number of bits per symbol ($\log_2{q}=3$). An ECC scheme should be applied in order to reduce the relatively high raw BER of $0.08$ to a BER of $10^{-15}$, which is required by the consumer. Let us consider using state of the art LDPC codes designed for flash memories presented in~\cite{LDPC} (Fig. 10). We can see that for a raw BER of $0.08$ none of the LDPC codes is useful. In fact, in order to reach a BER of $10^{-15}$ the raw BER of the strongest LDPC should be at most $0.013$. Let us now examine concatenating the LDPC codes with asymmetric magnitude-$1$ codes as inner codes. The asymmetric version of the BCH code reduces the raw BER of $0.08$ to $0.035$, which is not enough for the LDPC to converge. The even/odd scheme is slightly better, reducing the BER to $0.0203$, which is still insufficient. The NCC, however, reduces the raw BER of $0.08$ to $0.0065$, which can enable a soft LDPC with just two extra sensing levels to reach the required BER of $10^{-15}$.
\end{example}

\section{Conclusion and Future Work}
The NCC flexible coding scheme for asymmetric magnitude-$1$ errors was presented. An optimal, low complexity decoding algorithm was introduced. The NCC was also analyzed and compared to BCH and even/odd codes, outperforming these codes for moderate to high values of SER. Several bounds on the error-correcting performance were derived. The NCC was also analyzed by finite block length analysis tools, indicating that this scheme has the potential to be nearly optimal. \\
For future work, it is possible to enhance the flexibility of the constraint and to expand the rate/correcting capabilities tradeoff. Furthermore, in order to address high error propagation when concatenating the NCC with some outer code, a more sophisticated and efficient concatenation method can de developed. On the theoretical side, developing sphere-packing bounds for the $q$-ary Z-channel may suggest additional analyses and refinements of the NCC. In addition, aside of correcting errors, the NCC can also be used for speeding up read and write processes in non-volatile memories.
\section{Acknowledgment}
The authors wish to thank Y. Polyanskiy for his guidance on the finite block-length analysis. This work was supported by the Israel Science Foundation, by the Israel Ministry of Science and Technology, and by a GIF Young Investigator grant.


\appendices
\section{proof of Theorem~\ref{th:proofApp}} \label{app:A}
We prove Theorem~\ref{th:proofApp} for the common scenario in which $q$ is an even number. The proof for the odd $q$ scenario is straightforward. This result is proved with the help of the following lemmas.
\begin{definition}
A subcode of the NCC, $NCC_r\left(n,q\right)$ is a code containing all the codewords of $NCC\left(n,q\right)$ that occupy exactly $r$ levels. It is possible to define in a similar way the subcode $even/odd_r\left(n,q\right)$ of the even/odd code.
\end{definition}
\begin{lemma}
The information rates of $NCC_{q/2}\left(n,q\right)$ and $even/odd_{q/2}\left(n,q\right)$ are identical when $\cellblocksize\rightarrow \infty$.
\label{prop:ncc_eq_evenodd1}
\end{lemma}
\begin{proof}
From Theorem~\ref{th:NCCRate} we know that when the number of occupied levels in a codeword is $q/2$, which is the maximal value of occupied levels, there are $q/2+1$ combinations of levels for an NCC codeword, while there are only two for the even/odd code ($[0,2,\ldots,q-2]$ or $[1,3,\ldots,q-1]$). Let us denote by $M$ the number of $even/odd_{q/2}\left(n,q\right)$ codewords. The number of $NCC_{q/2}\left(n,q\right)$ codewords is therefore $\frac{q/2+1}{2}M$, so the information rate in this case is
\begin{equation}
\cR_{NCC}= log_q\left(\frac{q/2+1}{2}M\right)/n = \cR_{even/odd} + log_q\left(\frac{q/2+1}{2}\right)/n,
\end{equation}
Therefore, it is clear that when $q/2$ levels are occupied and $\cellblocksize\rightarrow \infty$ the information rates of the NCC and even/odd schemes are identical.
\end{proof}
\begin{lemma}
When we uniformly select an $NCC(\cellblocksize,q)$ codeword, the probability $Prob\left(\cellblocksize,q/2\right)$ for selecting a $NCC_{q/2}\left(n,q\right)$ codeword is asymptotically $1$ when $\cellblocksize\rightarrow\infty$. In other words, when $\cellblocksize$ is large, the subcode $NCC_{q/2}\left(n,q\right)$ and the $NCC\left(n,q\right)$ code coincide.
\label{prop:ncc_eq_evenodd2}
\end{lemma}
\begin{proof}
From the proof of Theorem~\ref{th:NCCRate} it is clear that the probability $Prob\left(\cellblocksize,q/2\right)$ to choose $\cellblocksize$-length, $q/2$ occupied levels codeword is given by:
\begin{equation}
Prob\left(\cellblocksize,\frac{q}{2}\right) = \frac{\frac{q}{2}!\cdot S\left(n,\frac{q}{2}\right) \cdot \left(\frac{q}{2}+1\right) }{\sum_{k=1}^{\frac{q}{2}} k!\cdot S\left(n,k\right) \cdot \binom{q-k+1}{k}}
\end{equation}
Recall that asymptotically $S\left(n,k\right)\sim \frac{k^n}{k!}+\cO\left((k-1)^n\right)$, therefore we get
\begin{equation}
Prob\left(\cellblocksize,\frac{q}{2}\right) \sim \frac{\left(\frac{q}{2}\right)^n \cdot \left(\frac{q}{2}+1\right) }{\sum_{k=1}^{\frac{q}{2}-1} k^n \cdot \binom{q-k+1}{k} + \left(\frac{q}{2}\right)^n \cdot \left(\frac{q}{2}+1\right)}
\end{equation}
After some algebra we get:
\begin{equation}
Prob\left(\cellblocksize,\frac{q}{2}\right) \sim \frac{1 }{\sum_{k=1}^{\frac{q}{2}-1} \left(\frac{2k}{q}\right)^n \cdot \binom{q-k+1}{k}/\left(\frac{q}{2}+1\right) + 1}
\end{equation}
Due to the fact that for the relevant values of $k$, we get $\frac{2k}{q}<1$ and that $\binom{q-k+1}{k}/\left(\frac{q}{2}+1\right)$ is bounded and does not depend on $n$ it is obvious that
\begin{equation}
\lim _{n\rightarrow\infty} Prob\left(\cellblocksize,\frac{q}{2}\right) =1
\end{equation}
\end{proof}
When $n\rightarrow\infty$ it is clear by~\eqref{eq:EvenOddrate} that $\cR_{even/odd} \rightarrow \cR_{all-even}$. From Lemmas~\ref{prop:ncc_eq_evenodd1} and~\ref{prop:ncc_eq_evenodd2} it is clear that the information rates of the NCC and even/odd are the same.

\end{document}